\documentclass{article}
\usepackage{kr}
\usepackage{microtype}

\usepackage{times}
\usepackage{soul}
\usepackage{url}
\usepackage[hidelinks]{hyperref}
\usepackage[utf8]{inputenc}
\usepackage[small]{caption}
\usepackage{graphicx}
\usepackage{amsmath}
\usepackage{amsthm}
\usepackage{booktabs}
\usepackage{algorithm}
\usepackage{algorithmic}
\usepackage[normalem]{ulem}

\newtheorem{df}{Definition}[section]  
\newtheorem{definition}{Definition}[section]  
\newtheorem{example}[df]{Example}
\newtheorem{theorem}[df]{Theorem}
\newtheorem{lemma}[df]{Lemma}

\newtheorem{proposition}[df]{Proposition}

\title{Strongly complete axiomatization for a logic\\ with probabilistic interventionist counterfactuals 
}

\author{%
Fausto Barbero$^1$\and
Jonni Virtema$^2$
\affiliations
$^1$University of Helsinki\\
$^2$University of Sheffield\\
\emails
fausto.barbero@helsinki.fi,
j.t.virtema@sheffield.ac.uk
}

\usepackage{amsmath}
\usepackage{amsthm}
\usepackage{latexsym}
\usepackage{amssymb}
\usepackage{verbatim}
\usepackage{pbox}
\usepackage{enumerate}
\usepackage{multicol}
\usepackage{xcolor}
\usepackage{hyperref}

\usepackage{graphicx}
\usepackage{tikz}
\usetikzlibrary{tikzmark}

{\usepackage{txfonts}  
   \DeclareSymbolFont{symbolsC}{U}{txsyc}{m}{n}
   \DeclareMathSymbol{\strictif}{\mathrel}{symbolsC}{74}
   \DeclareMathSymbol{\boxright}{\mathrel}{symbolsC}{128}
}

\newcommand{\cf}{\boxright}

\newcommand{\CO}{\mathcal{CO}}

\newcommand{\PCO}{\mathcal{PCO}}

\newcommand{\Al}{\hat \alpha}
\newcommand{\B}{\mathbb{B}}

\newcommand{\F}{\mathcal{F}}
\newcommand{\G}{\mathcal{G}}

\newcommand{\T}{\mathbb{T}}

\newcommand{\dfn}{\mathrel{\mathop:}=}
\newcommand{\dom}{\mathrm{Dom}}
\newcommand{\ran}{\mathrm{Ran}}
\newcommand{\End}{\mathrm{End}}
\newcommand{\PA}{\mathrm{PA}}

\newcommand{\SET}[1]{\mathbf{#1}}

\newcommand{\aff}[2]{{#1}\leadsto{#2}}

\usepackage[switch,mathlines]{lineno}
\usepackage{etoolbox}

\newcommand*\linenomathpatch[1]{%
  \cspreto{#1}{\linenomath}%
  \cspreto{#1*}{\linenomath}%
  \csappto{end#1}{\endlinenomath}%
  \csappto{end#1*}{\endlinenomath}%
}
\newcommand*\linenomathpatchAMS[1]{%
  \cspreto{#1}{\linenomathAMS}%
  \cspreto{#1*}{\linenomathAMS}%
  \csappto{end#1}{\endlinenomath}%
  \csappto{end#1*}{\endlinenomath}%
}

\expandafter\ifx\linenomath\linenomathWithnumbers
  \let\linenomathAMS\linenomathWithnumbers
  \patchcmd\linenomathAMS{\advance\postdisplaypenalty\linenopenalty}{}{}{}
\else
  \let\linenomathAMS\linenomathNonumbers
\fi

\linenomathpatch{equation}
\linenomathpatchAMS{gather}
\linenomathpatchAMS{multline}
\linenomathpatchAMS{align}
\linenomathpatchAMS{alignat}
\linenomathpatchAMS{flalign}

\makeatletter
\patchcmd{\mmeasure@}{\measuring@true}{
  \measuring@true
  \ifnum-\linenopenaltypar>\interdisplaylinepenalty
    \advance\interdisplaylinepenalty-\linenopenalty
  \fi
  }{}{}
\makeatother

\begin{document}

\maketitle

\begin{abstract}
Causal multiteam semantics is a framework where probabilistic notions and causal inference can be studied in a unified setting. We study a logic ($\PCO$) that features marginal probabilities and interventionist counterfactuals, and allows expressing conditional probability statements, $do$ expressions and other mixtures of causal and probabilistic reasoning.
Our main contribution is a strongly complete infinitary axiomatisation for $\PCO$.
%
\end{abstract}





\section{Introduction}

In the field of \emph{causal inference} \cite{Pea2000,SpiGlySch1993} it has become customary to express many notions of causation in terms of so-called \emph{interventionist counterfactuals}. In their simplest form, these are expressions such as:
\begin{center}
    If variables $X_1,\dots,X_n$ were fixed to values $x_1,\dots,x_n$, then variable $Y$ would take value $y$
\end{center}      
where the antecedent describes some kind of ``intervention'' on a given system.
In practice, in the field of causal inference one is interested in estimating probabilities of events. A wide variety of causal-probabilistic expressions appears in the literature (\emph{do} expressions, conditional \emph{do} expressions, and what J. Pearl calls ``counterfactuals''). Such expressions all involve probabilities in a post-intervention scenario, but differ in whether they involve probabilistic conditioning or not, and in whether one conditions upon events in the pre-intervention or the post-intervention scenario.  \cite{BarSan2018} proposed to decompose these kinds of expressions in terms of three simpler ingredients: marginal probabilities, interventionist counterfactuals, and selective implications. The selective implication describes the effect of acquiring new information whereas the interventionist counterfactual describes the effect of an action. The three operators can be studied in a shared semantic framework called \emph{causal multiteam semantics}. The framework is meaningful already in a non-probabilistic context, where it has been studied both from a semantic and a proof-theoretic perspective \cite{BarSan2020,BarYan2022,BarGal2022}. The proof-theoretical results rely on a body of earlier work (\cite{GalPea1998,Hal2000,Bri2012}) on proof systems for (non-probabilistic) counterfactuals evaluated on \emph{structural equation models} (of which the causal multiteams are a generalization). In the probabilistic setting, some work in the semantic direction is forthcoming \cite{BarSan2023,BarVir2023}. In the present paper, we address for the first time the issue of the proof theory of the probabilistic counterfactual languages that are considered in these papers. To the best of our knowledge, there has been only one proposal in the literature of a deduction system for probabilistic interventionist counterfactuals \cite{IbeIca2020}. The language considered in \cite{IbeIca2020} differs in many respects from those we are interested in; it is more expressive, in allowing the use of arithmetical operations (sums and products of probabilities and scalars), and less expressive, since it does not allow for nesting of counterfactuals (iterated interventions), nor it has any obvious way of describing complex interactions of interventions and conditioning -- such as conditioning at the same time on \emph{both} a pre-intervention and a post-intervention scenario, or conditioning on a state of affairs that holds at an intermediate stage between two interventions.

Axiomatizing probabilistic logics is a notoriously difficult problem. As soon as a language allows to express inequalities of the form $\Pr(\alpha) \leq \epsilon$ ($\epsilon$ being a rational number), it is not compact, as for example the set of formulas of the  form $\Pr(\alpha) \leq \frac{1}{n}$ ($n$ natural number) entails that  $\Pr(\alpha) = 0$, but no finite subset yields the same conclusion. Consequently, no usual, finitary deduction system can be strongly complete for such a language. A possible answer to this problem is to settle for a deductive system that is weakly complete, i.e. it captures all the correct inferences from \emph{finite} sets of formulas. This has been achieved for a variety of probabilistic languages with arithmetic operations (e.g. \cite{FagHalMeg1990}). 
The result for probabilistic interventionist counterfactuals mentioned above \cite{IbeIca2020} is a weak completeness result in this tradition. Proving weak completeness for probabilistic languages \emph{without} arithmetical operations seems to be a more difficult task, and we could find only one such result in the literature \cite{HeiMon2001}. Unfortunately, the completeness proof of \cite{HeiMon2001} relies on a model-building method that seems not to work for languages where \emph{conditional} probabilities are expressible; thus, it is not adaptable in any straightforward way  to our case. 

Another path, on which we embark, is to respond to the failure of compactness by aiming for strong completeness using a deduction system with some kind of infinitary resources. The use of infinitary deduction rules (with countably many premises) has proved to be very fruitful and has led to strong completeness theorems for a plethora of probabilistic languages (cf. \cite{OgnRasMar2016}). Of particular interest to us are \cite{RasOgnMar2004}, where  strong completeness is obtained for a language with conditional probabilities, and \cite{OgnPerRas2008}, which obtains strong completeness for ``qualitative probabilities'' (i.e., for expressions such as $\Pr(\alpha)\leq\Pr(\beta)$). We build on these works in order to obtain a strongly complete deduction system (with two infinitary rules) for the probabilistic-causal language $\PCO$ used in \cite{BarSan2023,BarVir2023}. The proof proceeds via a canonical model construction, which relies on a transfinite version of the Lindenbaum lemma. While the proof follows essentially the scheme of \cite{RasOgnMar2004}, it presents peculiar difficulties of its own due to the presence of additional operators (counterfactuals and comparison atoms).


\section{Preliminaries}\label{sec: preliminaries}

\subsection{Basic notation}

Capital letters such as $X, Y,\dots$ denote \textbf{variables} (thought to stand
for specific magnitudes such as ``temperature'', ``volume'', etc.) which
take \textbf{values} denoted by small letters (e.g. the values of the variable $X$ will be
denoted by $x, x', \dots$). 
Sets (and tuples, depending on the context) of variables and values are denoted by boldface letters such as $\SET X$ and $\SET x$.
We consider probabilities that arise from the counting measures of finite (multi)sets. For finite sets $S\subseteq T$, we define 
\(
P_T(S):= \frac{|S|}{|T|}.
\)

\subsection{Multiteams and causal multiteams}

A \textbf{signature} is a pair $(\dom,\ran)$, where $\dom$ is a nonempty, finite set of variables and $\ran$ is a function that associates to each variable $X\in \dom$ a nonempty, finite set $\ran(X)$ of values (the \textbf{range} of $X$). 
We consider throughout the paper a fixed ordering of $\dom$, and write $\SET W$ for the tuple of all variables of $\dom$ listed in such order. Furthermore, we write $\SET W_X$ for the variables of $\dom\setminus\{X\}$ listed according to the fixed order. 
Given a tuple $\SET X = (X_1,\dots, X_n)$ of variables, we denote as $\ran(\SET X)$ the Cartesian product $\ran(X_1)\times\dots\times \ran(X_n)$. 
An \textbf{assignment} of signature $\sigma$ is a mapping $s:\dom\rightarrow\bigcup_{X\in \dom}\ran(X)$ such that $s(X)\in \ran(X)$
for each $X\in \dom$.
The set of all assignments of signature $\sigma$ is denoted by $\B_\sigma$.
Given an assignment $s$ that has the variables of $\SET X$ in its domain, $s(\SET X)$ will denote the tuple $(s(X_1), \dots,s(X_n))$.
For $\SET X\subseteq \dom$, $s_{\upharpoonright \SET X}$ denotes the restriction of $s$ to the variables in $\SET X$.

A \textbf{team} $T$ of signature $\sigma$ is a subset of $\B_\sigma$.
Intuitively, a multiteam is just a multiset analogue of a team.
%
%
%
We represent \textbf{multiteams} as (finite) sets of assignments with an extra variable $Key$ (not belonging to the signature) ranging over $\mathbb N$, which takes different values over different assignments of the team, and which is never mentioned in the formal languages. 
 A multiteam can be represented as a table, in which each row represents an assignment. For example, if $Dom = \{X,Y,Z\}$, a multiteam may look like this: 

\begin{center}
\begin{tabular}{|c|c|c|c|}
\hline
\multicolumn{4}{|l|}{\small{Key} \ \ \ \ \ $X$ \ \ \ \ \ $Y$ \ \ \ \ \ $Z$} \\
\hline
 $\phantom{a}$0$\phantom{a}$ & $x$ & $y$ & $z$\\
\hline
 1 & $x'$ & $y'$ & $z'$ \\
\hline
 2 & $x'$ & $y'$ & $z'$ \\
\hline
\end{tabular}
\end{center}

%

Multiteams by themselves do not encode any solid notion of causation; they do not tell us how a system would be affected by an intervention. We therefore need to enrich multiteams with additional structure.

\begin{definition}\label{def: causal multiteam}
A \textbf{causal multiteam} $T$ of signature $(\dom(T), \ran(T))$ 
with endogenous variables $\mathbf V\subseteq \dom(T)$ is a pair $T = (T^-,\F)$, where
\begin{enumerate}
\item 
$T^-$ is a multiteam of domain $\dom(T)$,
and
\item
$\F$ is a function $\{(V,\F_V) \ | \ V\in\mathbf V\}$ that assigns to each endogenous variable $V$ a non-constant $|\SET W_V|$-ary function $\F_V: \ran(\SET W_V)\rightarrow 
\ran(V)$,
which satisfies the further \textbf{compatibility constraint}: 
\begin{center}
$\F_V(s(\SET W_V))=s(V)$, for all $s\in T^-$.
\end{center}
%
\end{enumerate}
\end{definition}
%
%
%
%
\noindent We will also write $\End(T)$ for the set of endogenous variables of $T$. 

The function $\mathcal F$ induces a system of structural equations; an equation
\[
V := \mathcal F_V(\SET W_V)
\]
for each variable $V\in End(T)$. A structural equation tells how the value of $V$ should be recomputed if the value of some variables in $\SET W_V$ is modified. Note that that some of the variables in $\SET W_V$ may not  be necessary for evaluating $V$. For example, if $V$ is given by the structural equation $V:= X+1$, all the variables in $\SET W_V\setminus \{X\}$ are irrelevant (we call them \textbf{dummy arguments} of $\F_V$). The set of non-dummy arguments of $\F_V$ is denoted as $\PA_V$ (the set of \textbf{parents} of $V$).

We associate to each causal multiteam $T$ a \textbf{causal graph} $G_T$, whose vertices are the variables in $\dom$ and where an arrow is drawn from each variable in $\PA_V$ to $V$, whenever $V\in \End(T)$. 
The variables in $\dom(T)\setminus \End(T)$ are called \textbf{exogenous}.



\begin{definition}
A causal multiteam $S=(S^-,\F_S)$ is a \textbf{causal sub-multiteam} of $T=(T^-,\F_T)$, if they have same signature, $S^-\subseteq T^-$, and $\mathcal{F}_S = \mathcal{F}_T$. We then write $S\leq T$.
%
\end{definition}


We consider causal multiteams as dynamic models, that can be affected by various kinds of operations -- specifically, by observations and interventions.
Given a causal multiteam $T = (T^-,\F)$ and a formula $\alpha$ of some formal language (evaluated over assignments according to some semantic relation $\models$), ``observing $\alpha$'' produces the causal sub-multiteam $T^\alpha = ((T^\alpha)^-,\F)$ of $T$, where
\[
(T^\alpha)^- \dfn \{s\in T^-  \ | \  (\{s\},\F)\models \alpha\}.
\]

An intervention on $T$ will \emph{not}, in general, produce a sub-multiteam of $T$. It will instead modify the values that appear in some of the columns of $T$. We consider interventions that are described by formulas of the form $X_1=x_1 \land\dots\land X_n=x_n$ (or, shortly, $\SET X = \SET x$). Such a formula is \textbf{inconsistent} if there are two indexes $i,j$ such that $X_i$ and $X_j$ denote the same variable, while $x_i$ and $x_j$ denote distinct values; it is \textbf{consistent} otherwise.

Applying an intervention $do(\SET X = \SET x)$, where $\SET X = \SET x$ is consistent, to a causal multiteam $T = (T^-,\F)$ will produce a causal multiteam $T_{\SET X = \SET x} = (T^-_{\SET X = \SET x},\F_{\SET X = \SET x})$, where the function component is 
\(
\F_{\SET X = \SET x} := \F_{\upharpoonright(\SET V \setminus \SET X)}
\)
(the restriction of $\F$ to the set of variables $\SET V \setminus \SET X$) 
and the multiteam component is 
\(
T_{\SET X=\SET x}^-:=\{s^\F_{\SET X=\SET x}\mid s\in T^-\},
\)
where each $s^\F_{\SET X=\SET x}$ is the unique assignment compatible with $\mathcal F_{{\SET X = \SET x}}$ defined (recursively) as 
\[s^\F_{\SET X=\SET x}(V)=\begin{cases}
x_i&\text{ if }V=X_i\in \SET X\\
s(V)&\text{ if }V\in Exo(T)\setminus \SET X\\
\F_V(s^\F_{\SET X=\SET x}(\PA_V
))&\text{ if }V\in End(T)\setminus \SET X.
\end{cases}
\]


\begin{example}
Consider the following table:
\begin{center}
$T^-$: \begin{tabular}{|c|c|c|c|}
\hline
\multicolumn{4}{|c|}{ } \\
\multicolumn{4}{|l|}{\small{Key} \ \ $X$\tikzmark{XR} \ \ \ \tikzmark{YL}$Y$\tikzmark{YR} \ \ \ \tikzmark{ZL}$Z$} \\
\hline
 0 & 0 & 1 & 0\\
\hline
 1 & 1 & 2 & 2 \\
\hline
 2 & 1 & 2 & 2 \\
\hline
 3 & 2 & 3 & 6\\ 
\hline
\end{tabular}
\begin{tikzpicture}[overlay, remember picture, yshift=0\baselineskip, shorten >=.5pt, shorten <=.5pt]
  \draw ([yshift=7pt]{pic cs:XR})  edge[line width=0.2mm, out=45,in=135,->] ([yshift=7pt]{pic cs:ZL});
  \draw [->] ([yshift=3pt]{pic cs:YR})  [line width=0.2mm] to ([yshift=3pt]{pic cs:ZL});
  \draw [->] ([yshift=3pt]{pic cs:XR})  [line width=0.2mm] to ([yshift=3pt]{pic cs:YL});
\end{tikzpicture}
\end{center}
where each row represents an assignment (e.g., the fourth row represents an assignment $s$ with $s(Key)=3$, $s(X)=2$, $s(Y)=3$, $s(Z)=6$). Assume further that the variable $Z$ is generated by the function $\F_Z(X,Y) = X \times Y$, $Y$ is generated by $\F_Y(X)=X+1$, and $X$ is exogenous. The rows of the table are compatible with these laws, so this is a causal multiteam (call it $T$). It encodes many probabilities; for example, $P_T(Z=2) = \frac{1}{2}$. Suppose we have a way to enforce the variable $Y$ to take the value $1$. We represent the effect of such an intervention ($do(Y=1)$) by recomputing the $Y$ and then  the $Z$ column:

\begin{center}
 \begin{tabular}{|c|c|c|c|}
\hline
\multicolumn{4}{|c|}{ } \\
\multicolumn{4}{|l|}{ \hspace{-5pt} \small{Key} \ \ $X$\tikzmark{XR'} \ \ \ \tikzmark{YL'}$Y$\tikzmark{YR'} \ \ \ \ \ \tikzmark{ZL'}$Z$} \\
\hline
 0 & 0 & \textbf{1} & $\dots$ \\
\hline
 1 & 1 & \textbf{1} &  $\dots$  \\
\hline
 2 & 1 & \textbf{1} &  $\dots$ \\
\hline
 3 & 2 & \textbf{1} &  $\dots$ \\ 
\hline
\end{tabular}
\hspace{15pt} $\leadsto$ \hspace{5pt}  $T_{Y=1}^-$: \begin{tabular}{|c|c|c|c|}
\hline
\multicolumn{4}{|c|}{ } \\
\multicolumn{4}{|l|}{\small{Key} \  $X$\tikzmark{XR''} \ \ \ \tikzmark{YL''}$Y$\tikzmark{YR''} \ \ \ \ \tikzmark{ZL''}$Z$} \\
\hline
 0 & 0 & 1 & \textbf{0}\\
\hline
 1 & 1 & 1 & \textbf{1} \\
\hline
 2 & 1 & 1 & \textbf{1} \\
\hline
 3 & 2 & 1 & \textbf{2}\\ 
\hline
\end{tabular}
\begin{tikzpicture}[overlay, remember picture, yshift=0\baselineskip, shorten >=.5pt, shorten <=.5pt]
  \draw ([yshift=7pt]{pic cs:XR'})  edge[line width=0.2mm, out=45,in=135,->] ([yshift=7pt]{pic cs:ZL'});
  \draw [->] ([yshift=3pt]{pic cs:YR'})  [line width=0.2mm] to ([yshift=3pt]{pic cs:ZL'});
  
  \draw ([yshift=7pt]{pic cs:XR''})  edge[line width=0.2mm, out=45,in=135,->] ([yshift=7pt]{pic cs:ZL''});
  \draw [->] ([yshift=3pt]{pic cs:YR''})  [line width=0.2mm] to ([yshift=3pt]{pic cs:ZL''});
\end{tikzpicture}

\end{center}
where the new value of $Z$ is computed, in each row, as the product of the value for $X$ and the (new) value for $Y$. 
The probability distribution has changed: now  $P_{T_{Y=1}}(Z=2)$ is $\frac{1}{4}$. Furthermore, the function $\F_Y$ is now omitted from $T_{Y=1}$ (otherwise the assignments would not be compatible anymore with the laws). Correspondingly, the arrow from $X$ to $Y$ has been omitted from the causal graph.   
\end{example}

\subsection{Language $\CO$}

The language $\CO$ is for the description of events; later we incorporate it in a language for the discussion of probabilities of $\CO$ formulas. For any fixed signature, the formulas of $\CO$ are defined by the following BNF grammar:
\[
\alpha ::= Y=y \mid Y\neq y \mid  \alpha\land\alpha  \mid \alpha\supset\alpha  \mid  \SET X = \SET x \cf \alpha,
\]
where $\SET{X}\cup\{Y\}\subseteq \dom$, $y\in \ran(Y)$, and $\SET x\in \ran(\SET X)$.
Formulae of the forms $Y = y$ and $Y\neq y$ are called \textbf{literals}.
The semantics for $\CO$ is given by the following clauses:
\begin{align*}
&T\models Y=y &&\text{iff}&& s(Y)=y \text{ for all } s\in T^-.\\
&T\models Y\neq y  &&\text{iff}&& s(Y)\neq y \text{ for all }s\in T^-.\\
&T\models \alpha\land \beta  &&\text{iff}&& T\models \alpha \text{ and } T\models \beta.\\
&T\models \alpha \supset \beta &&\text{iff} && T^\alpha \models \beta.\\
&T\models \SET X=\SET x \cf \psi &&\text{iff}&& T_{\SET X=\SET x} \models \psi \text{ or } \SET X=\SET x \text{ is inconsistent}.
\end{align*}
We will reserve the letters $\alpha,\beta$ to denote $\CO$ formulas.

We can introduce more logical operators as useful abbreviations. $\top$ stands for $X=x \cf X=x$, and 
$\bot$ stands for $X=x \cf X\neq x$. 
$\neg\alpha$ (\emph{dual negation}) stands for $\alpha\supset \bot$. This is not a classical (contradictory) negation; it is easy to see that its semantics is:
\begin{itemize}
\item $(T^-,\F)\models \neg\alpha$ iff, for every $s\in T^-$, $(\{s\},\F)\not\models\alpha$.
\end{itemize}
Thus, it is not the case, in general, that $T\models\alpha$ or $T\models \neg\alpha$. Note that $X\neq x$ is semantically equivalent to $\neg (X=x)$, and $X = x$ is semantically equivalent to $\neg (X\neq x)$.
In previous works $\lor$ (\emph{tensor disjunction}) was taken as a primitive operator, but here we define 
 $\alpha\lor\beta$ as $\neg(\neg\alpha \land \neg\beta)$.
Its semantic clause can be described as follows:
\begin{itemize}
\item $T\models \alpha\lor \beta$ iff there are $T_1,T_2\leq T$ s.t. $ T_1^-\cup T_2^- = T^-$,      $T_1\models \alpha$ and  $T_2\models \beta$.
\end{itemize}
In contrast with the statement above, the formula $\alpha\lor\neg\alpha$ \emph{is} valid. Furthermore, $\alpha\equiv\beta$ abbreviates $(\alpha \supset \beta) \land (\beta\supset \alpha)$. 
Notice that this formula does not state that $\alpha$ and $\beta$ are logically equivalent, but only that they are satisfied by the same assignments in the specific causal multiteam at hand.

All the operators discussed here (primitive and defined) behave classically over causal multiteams containing exactly one assignment.

A causal multiteam $(T^-,\F)$ is \textbf{empty} (resp. \textbf{nonempty}) if the multiteam $T^-$ is.
All the logics $\mathcal{L}$ considered in the paper have
the \textbf{empty team property}: if $T$ is empty, then $T\models\alpha$ for any $\alpha\in\mathcal{L}$ (and any $\F$ of the same signature). Furthermore, $\CO$ is \textbf{flat}: $(T^-,\F)\models \alpha$ iff, for all $s\in T^-$, $(\{s\},\F)\models\alpha$.

\subsection{Language $\PCO$}

Our main object of study is the probabilistic language $\PCO$. Besides literals, it allows for \textbf{probabilistic atoms}:
\[
\ \Pr(\alpha) \geq \epsilon \ | \ \Pr(\alpha) > \epsilon \ | \ \Pr(\alpha) \geq \Pr(\beta) \ | \ \Pr(\alpha) > \Pr(\beta)
\]
where $\alpha,\beta\in\CO$ and $\epsilon \in [0,1]\cap \mathbb{Q}$. The first two are called \textbf{evaluation atoms}, and the latter two \textbf{comparison atoms}. 
Probabilistic atoms 
 together with literals of $\CO$ are called \textbf{atomic formulas}.
The probabilistic language $\PCO$ is then given by the following grammar:
\[
\varphi::= \eta \mid \varphi \land \varphi \mid  \varphi \sqcup \varphi \mid \alpha\supset \varphi \mid \SET X = \SET x\cf \varphi,
\]
where $\SET X\subseteq \dom$, $\SET x \in \ran(\SET X)$, $\eta$ is an atomic formula, 
and $\alpha$ is a $\CO$ formula.
Note that the antecedents of $\supset$ and the arguments of probability operators are $\CO$ formulas.
%
Semantics for the additional operators are given below:
\begin{align*}
&T\models \psi\sqcup \chi &&\text{iff}&& T\models \psi \text{ or } T\models \chi\\
&T\models \Pr(\alpha)\vartriangleright \epsilon  &&\text{iff}&& T^-=\emptyset \text{ or } P_T(\alpha)\vartriangleright \epsilon\\
&T\models \Pr(\alpha)\vartriangleright \Pr(\beta) &&\text{iff}&& T^-=\emptyset \text{ or } P_T(\alpha)\vartriangleright P_T(\beta)
\end{align*}
where $\vartriangleright\in\{\geq,>\}$ and $P_T(\alpha)$ is a shorthand for $P_{T^-}((T^\alpha)^-)$.

As usual, for a set of formulas $\Gamma$, we write $T\models \Gamma$ if $T$ satisfies each of the formulas in $\Gamma$. For $\Gamma \cup \{\varphi\}\subseteq \PCO$, we write $\Gamma\models_\sigma\varphi$ if $T\models \Gamma$ implies $T\models\varphi$, for all causal teams $T$ of signature $\sigma$. $\models_\sigma\varphi$ abbreviates $\emptyset \models_\sigma \varphi$. We will always assume that some signature is fixed, and omit the subscripts. 

The language $\PCO$ still has the empty team property 
but it is not flat.\footnote{It also lacks weaker properties considered in the team semantics literature, such as \emph{downward closure} and \emph{union closure}.}

The abbreviations $\top,\bot$ can be used freely in $\PCO$, while $\neg,\lor$ and $\equiv$ can be applied only to $\CO$ arguments. The definability of the dual negation in $\CO$ allows us to introduce more useful abbreviations:
\begin{align*}
\Pr(\alpha) \leq  \epsilon &\dfn\Pr(\neg\alpha) \geq 1 - \epsilon \\ 
\Pr(\alpha) < \epsilon &\dfn \Pr(\neg\alpha) > 1 - \epsilon\\
\Pr(\alpha) = \epsilon &\dfn \Pr(\alpha) \geq  \epsilon \land\Pr(\alpha) \leq  \epsilon\\
\Pr(\alpha) \neq  \epsilon &\dfn \Pr(\alpha) >  \epsilon \sqcup  \Pr(\alpha) <  \epsilon.
\end{align*}
Furthermore, the $\supset$ operator 
enables us to express some statements involving conditional probabilities, defined as follows (where  $\vartriangleright \hspace{3pt} \in \{\geq, >\}$): 
\begin{align*}
&T\models \Pr(\alpha\mid\gamma) \vartriangleright \epsilon  &&\hspace{-5pt}\text{iff}  &&\hspace{-5pt}(T^\gamma)^- = \emptyset \text{ or } P_{T^\gamma}(\alpha) \vartriangleright \epsilon.\\
&T\models \Pr(\alpha\mid\gamma)\vartriangleright \Pr(\beta\mid\gamma) &&\hspace{-5pt}\text{iff} &&\hspace{-5pt}(T^\gamma)^- = \emptyset \text{ or } P_{T^\gamma}(\alpha) \vartriangleright P_{T^\gamma}(\beta).
\end{align*}
It was observed in \cite{BarSan2023} that $\Pr(\alpha\mid\gamma) \vartriangleright \epsilon$ and $\Pr(\alpha\mid\gamma)\vartriangleright \Pr(\beta\mid\gamma)$ can be defined by $\gamma \supset\Pr(\alpha) \vartriangleright \epsilon$ and $\gamma \supset\Pr(\alpha) \vartriangleright \Pr(\beta)$, respectively.

The \emph{weak contradictory negation} $\varphi^C$ of a formula $\varphi$ is inductively definable in $\PCO$; this is an operator that behaves exactly as a contradictory negation, except on empty causal multiteams. 
We list the definitory clauses together with the values produced by the negation of defined formulas.

\begin{itemize}
\item $(\Pr(\alpha)\geq \epsilon)^C$ is $\Pr(\alpha) < \epsilon$ (and vice versa)  
\item $(\Pr(\alpha) > \epsilon)^C$ is $\Pr(\alpha)\leq \epsilon$ (and vice versa)
\item $(\Pr(\alpha) = \epsilon)^C$ is $\Pr(\alpha)\neq \epsilon$ (and vice versa)
\item $(\Pr(\alpha)\geq \Pr(\beta))^C$ is $\Pr(\beta) > \Pr(\alpha)$  (and vice versa)
\item $(\bot)^C$ is $\top$ (and vice versa)
\item $(X=x)^C$ is $\Pr(X=x) <1$
\item $(X\neq x)^C$ is $\Pr(X\neq x) <1$
\item $(\psi \land \chi)^C$ is $\psi^C \sqcup \chi^C$
\item $(\psi \sqcup \chi)^C$ is $\psi^C \land \chi^C$
\item $(\alpha \supset \chi)^C$ is $\Pr(\alpha)>0 \land \alpha \supset \chi^C$
\item $(\SET X = \SET x \cf \chi)^C$ is $\SET X = \SET x \cf \chi^C$.
\end{itemize}

\noindent In the clause for $\supset$, the conjunct $\Pr(\alpha)>0$  (whose intuitive interpretation is ``if $T$ is nonempty, then $T^\alpha$ is nonempty'') is added to insure that $(\alpha \supset \chi)^C$ is not satisfied by $T$ in case ($T$ is nonempty and) $T^\alpha$ is empty.\footnote{Whereas $\Pr(\alpha)>0$ could be replaced with $(\neg\alpha)^C$, the use of probability atoms in $(X=x)^C$ and $(X\neq x)^C$ seems essential.} 

We emphasise that, since $\CO$ formulas are $\PCO$ formulas, the weak contradictory negation can also be applied to them; however, the contradictory negation of a $\CO$ formula will typically \emph{not} be itself a $\CO$ formula.
The meaning of the weak contradictory negation is as follows.

\begin{theorem}\label{thm: contradictory negation}
For every $\varphi\in\PCO_{\sigma}$ and nonempty causal multiteam $T = (T^-,\F)$ of signature $\sigma$,
$T\models \varphi^C \Leftrightarrow T\not\models \varphi$.
\end{theorem}

\begin{proof}
The proof proceeds by induction on the structure of formulas 
 $\varphi$.
 We show the only non-trivial case of $\supset$.


Suppose $T\models 
\Pr(\alpha)>0  \land \alpha \supset \chi^C$. Thus $T^\alpha\models \chi^C$. 
Since $T$ is nonempty and $T\models \Pr(\alpha)>0$, we conclude that $T^\alpha$ is nonempty as well. Now by applying the induction hypothesis on $\chi$, we obtain $T^\alpha\not \models \chi$. Thus, $T\not\models \alpha\supset\chi$.

For the converse, assume $T\not\models \alpha\supset\chi$. Then $T^\alpha\not \models \chi$, which (by the empty team property) entails that $T^\alpha$ is nonempty, and thus $T\models \Pr(\alpha)>0$. Moreover, applying the induction hypothesis to $\chi$ yields $T^\alpha \models \chi^C$, and thus $T\models \alpha \supset \chi^C$. 
\end{proof}

Using the weak contradictory negation, we can define an operator that behaves exactly as the material conditional: 
\begin{itemize}
\item $\psi\rightarrow \chi$ stands for $\psi^C \sqcup \chi$.
\end{itemize}
Indeed, $T\models \psi\rightarrow \chi$ iff $T$ is empty or $T\not\models\psi$ or $T\models\chi$. However, since $\PCO$ has the empty multiteam property, ``$T$ is empty'' entails $T\models\chi$; thus, for $\PCO$, $\rightarrow$ really is the material conditional:

\begin{itemize}
\item $\psi\rightarrow \chi$ iff $T\not\models\psi$ or $T\models\chi$.
\end{itemize}
Similarly, we let $\psi\leftrightarrow\chi$ denote $(\psi\rightarrow \chi)\land(\chi\rightarrow \psi)$.

Note that $\alpha\rightarrow\beta$  
and  $\alpha\supset\beta$ are not in general equivalent even if $\alpha,\beta$ are $\CO$ formulas. Consider for example a causal multiteam $T$ with two assignments $s= \{(X,0),(Y,0)\}$ and $t= \{(X,1),(Y,1)\}$. Clearly $T\models X=0 \rightarrow Y=1$ (since $T\not\models X=0$), while $T\not\models X=0 \supset Y=1$ (since $T^{X=0}\not\models Y=1$).
However, the entailment from $\alpha\supset\psi$ to $\alpha\rightarrow \psi$ always holds, provided both formulas are in $\PCO$ (i.e., provided $\alpha\in\CO$). Indeed, suppose $T\models\alpha\supset\psi$ and $T\models\alpha$. From the former we get $T^\alpha \models \psi$. From the latter we get $T=T^\alpha$. Thus, $T\models\psi$. 
The opposite direction does not preserve truth, but it does preserve validity: if $\models \alpha\rightarrow \psi$, then $\models \alpha\supset\psi$. Indeed, the former tells us that any causal multiteam that satisfies $\alpha$ also satisfies $\psi$. Thus, in particular, for any $T$, $T^\alpha\models \psi$, and thus $T\models \alpha\supset \psi$.

Similar considerations as above apply to the pair of operators $\equiv$ and $\leftrightarrow$. 
Futher differences in the proof-theoretical behaviour of these (and other) pairs of operators are illustrated by the axioms T1 and T2 presented in Section \ref{sec: Axioms}.






\section{The axiom system}\label{sec: axiom system}

We present a formal deduction system with infinitary rules for $\PCO$ and show it to be strongly complete over recursive causal multiteams. 
 We follow the approach of \cite{RasOgnMar2004}, which proved a similar result for a language with probabilities and conditional probabilities. Our result adds to the picture comparison atoms, counterfactuals, and pre-intervention observations (``Pearl's counterfactuals'').

\subsection{Further notation}\label{subs: further notation}


The formulation of some of the axioms -- in particular, those involving reasoning with counterfactuals -- will involve some additional abbreviations. For example, we will write $\SET X \neq \SET x$ for a disjunction $X_1\neq x_1 \sqcup\dots\sqcup X_n\neq x_n$.



There will be an axiom (C11) that characterizes recursivity as done in \cite{Hal2000}.
For it, we need to define the atom $\aff{X}{Y}$ (``$X$ causally affects $Y$'') by the formula: 
\begin{equation*}
\hspace{-50pt}\bigvee_{\substack{\SET Z\subseteq Dom \\ x\neq x'\in \ran(X)\\ y\neq y'\in \ran(Y) \\ \SET z \in \ran(\SET Z)}} [((\SET Z=\SET z \land X=x)\cf Y=y) 
\end{equation*}
\vspace{-35pt}
\begin{equation*}
\hspace{60pt}\land ((\SET Z=\SET z  \land X=x')\cf Y=y')].
\end{equation*}

\vspace{12pt}

We will also need a formula (from \cite{BarSan2020}) characterizing the stricter notion of direct cause ($X$ is a direct cause of $Y$ iff $X\in PA_Y$), which is expressible by a $\PCO$ formula $\varphi_{DC(X,Y)}$ defined as:
\begin{equation*}
\hspace{-50pt}\bigvee_{\substack{x\neq x'\in \ran(X)\\ y\neq y'\in \ran(Y) \\ \SET w \in \ran(\SET W_{XY})}} [((\SET W_{XY}=\SET w \land X=x)\cf Y=y) 
\end{equation*}
\vspace{-30pt}
\begin{equation*}
\hspace{60pt}\land ((\SET W_{XY}=\SET w  \land X=x')\cf Y=y')].
\end{equation*}

\vspace{10pt}

\noindent where $\SET W_{XY}$ stands for $\dom\setminus\{X,Y\}$.


Fourth, some axioms describe specific properties of exogenous or endogenous variables, which can be again characterized in $\PCO$. Similarly as before, here  
 $\SET W_{V}$ stands for $\dom\setminus\{V\}$. 
Then we can express the fact that a variable $Y$ is endogenous by the following formula:

\vspace{-5pt}

\[
\varphi_{End(Y)}: \bigsqcup_{X\in \SET W_Y} \varphi_{DC(X,Y)}
\]

\vspace{-5pt}

\noindent and its contradictory negation $(\varphi_{End(Y)})^C$ will express that $Y$ is exogenous. 



Finally, for each function component $\F$, $\Phi^\F$ is a formula 
 that characterizes the fact that a causal team has function component $\F$. 
In detail, 
\[
\Phi^\F: \bigwedge_{V\in End(\F)} \eta_\sigma(V) \land \bigwedge_{V\notin End(\mathcal F)  
} \xi_\sigma(V)
\]
where
\[
\eta_\sigma(V): \bigwedge_{\SET w\in \ran(\SET W_V)}(\SET W_V = \SET w \cf V = \F_V(\SET w))
\]
and
\[
\xi_\sigma(V): \bigwedge_{\substack{\SET w\in \ran(\SET W_V) \\ v \in \ran(V)}} V=v \supset (\SET W_V=\SET w \cf V=v).
\]
A nonempty causal multiteam $T=(T^-,\G)$ satisfies $\Phi^\F$ iff $\G=\F$.\footnote{Save for some inessential differences, this is is the content of Theorem 3.4 from \cite{BarYan2022},}


\subsection{Axioms and rules}\label{sec: Axioms}

We present a few axiom schemes and rules for $\PCO$, roughly divided in six groups. Each axiom and rule is restricted to formulas of a fixed signature $\sigma$, so that actually we have a distinct axiom system for each signature. As usual, $\alpha$ and $\beta$ are restricted to be $\CO$ formulas.



\vspace{10pt}

\textbf{Tautologies}

\noindent T1. All instances of classical propositional tautologies in 

\hspace{8pt} $\land,\sqcup,\rightarrow, ^C, \top,\bot$.

\noindent T2. All $\CO$ instances of classical propositional tautologies  

\nopagebreak
\hspace{8pt} in $\land, \lor,\supset, \neg,\top,\bot$.







\vspace{3pt}

\noindent Rule MP. \LARGE$\frac{\psi \hspace{10pt} \psi\rightarrow\chi}{\chi}$\normalsize

\noindent Rule Rep. \LARGE$\frac{\vdash \varphi \hspace{15pt}  \vdash \theta \leftrightarrow \theta'}{\vdash \varphi[\theta'/\theta]}$\normalsize (provided $\varphi[\theta'/\theta]$ is well-formed)

\vspace{20pt}




\vspace{10pt}

\textbf{Probabilities}

\noindent P1. $\alpha \leftrightarrow \Pr(\alpha)=1$. 




\noindent P2. $\Pr(\alpha)\geq 0$.

\noindent P3. $(\Pr(\alpha)=\delta \land \Pr(\beta)=\epsilon \land \Pr(\alpha\land\beta)=0 )\rightarrow \Pr(\alpha\lor\beta)= \delta+\epsilon$

\hspace{8pt}(when $\delta+\epsilon \leq 1$).

\noindent P3b. $\Pr(\alpha) \geq \epsilon \land \Pr(\alpha\land\beta)=0 \rightarrow \Pr(\beta)\leq 1-\epsilon$.




\noindent P4. $\Pr(\alpha)\leq\epsilon \rightarrow \Pr(\alpha)< \delta$ (if $\delta > \epsilon$).

\noindent P5. $\Pr(\alpha)<\epsilon \rightarrow \Pr(\alpha)\leq\epsilon$.

\noindent P6. 
$\Pr(\alpha\equiv\beta)=1 \rightarrow (\Pr(\alpha)=\epsilon\rightarrow \Pr(\beta)=\epsilon)$. 

\noindent P6b. $\Pr(\alpha\supset\beta)=1 \rightarrow (\Pr(\alpha)=\epsilon\rightarrow \Pr(\beta)\geq\epsilon)$. 






\vspace{5pt}

\noindent Rule $\bot^\omega$. \LARGE$\frac{\psi \rightarrow \Pr(\alpha)\neq \epsilon, \forall \epsilon\in [0,1]\cap\mathbb Q}{\psi \rightarrow \bot}$\normalsize

\vspace{10pt}

\textbf{Comparison}

\noindent CP1. $(\Pr(\alpha)=\delta \land \Pr(\beta)=\epsilon)\rightarrow \Pr(\alpha)\geq \Pr(\beta)$, (if $\delta\geq \epsilon$).


\noindent CP2. $(\Pr(\alpha)=\delta \land \Pr(\beta)=\epsilon)\rightarrow \Pr(\alpha)> \Pr(\beta)$, (if $\delta > \epsilon$).




\vspace{10pt}

\textbf{Observations}

\noindent O1. $\Pr(\alpha)=0 \rightarrow (\alpha\supset\psi)$.


\noindent O1b. $(\alpha\supset \bot) \rightarrow \Pr(\alpha)=0$.

\noindent O2. $(\Pr(\alpha)=\delta \land \Pr(\alpha\land\beta)=\epsilon) \rightarrow (\alpha \supset \Pr(\beta)=\frac{\epsilon}{\delta})$

\hspace{8pt}(when $\delta\neq 0$).

\noindent O3. $(\alpha \supset \Pr(\beta)=\epsilon)\rightarrow (\Pr(\alpha)=\delta \leftrightarrow \Pr(\alpha\land\beta)=\epsilon\cdot\delta)$ 

\hspace{8pt}(when $\epsilon\neq 0$).









\noindent O4. $(\alpha \supset \psi) \rightarrow (\alpha \rightarrow \psi)$.




\noindent  
O5$_\land$. $\alpha \supset (\psi\land\chi)  \leftrightarrow (\alpha \supset \psi)\land (\alpha \supset \chi)$. 

\noindent O5$_\sqcup$. $\alpha \supset (\psi\sqcup\chi)   \leftrightarrow  (\alpha \supset \psi)\sqcup (\alpha \supset \chi)$. 

\noindent O5$_\supset$. $\alpha \supset (\beta\supset\chi)   \leftrightarrow  (\alpha \land \beta)\supset \chi$.

\vspace{5pt}

\noindent Rule Mon$_\supset$. \LARGE$\frac{\vdash \psi\rightarrow \chi}{\vdash (\alpha \supset \psi) \rightarrow (\alpha \supset \chi)}$\normalsize 

\vspace{5pt}

\noindent Rule $\rightarrow$to$\supset$. \LARGE$\frac{\vdash \alpha \rightarrow \psi}{\vdash \alpha \supset\psi}$\normalsize 

\vspace{5pt}

\noindent Rule $\supset^\omega$. \LARGE$\frac{\psi \rightarrow (\Pr(\alpha\land\beta)=\delta\epsilon \leftrightarrow \Pr(\alpha)=\epsilon), \forall \epsilon\in (0,1]\cap\mathbb Q}{\psi \rightarrow (\alpha \supset \Pr(\beta) = \delta)}$\normalsize

\vspace{5pt}




\vspace{10pt}

\textbf{Literals}

\noindent A1. $\SET Y=\SET y \rightarrow \SET Y\neq \SET y'$ (when $\SET y\neq \SET y'$).

\noindent A2. $X\neq x \leftrightarrow (X = x \supset \bot)$.

\noindent A3. $\bigvee_{\SET y\in \ran(\SET Y)} \SET Y = \SET y$.


\vspace{10pt}

\textbf{Counterfactuals}

\noindent C1. $(\SET X = \SET x \cf (\psi\land\chi)) \leftrightarrow ((\SET X = \SET x \cf \psi) \land (\SET X = \SET x \cf \chi))$.

\noindent C2. $(\SET X = \SET x \cf (\psi\sqcup\chi)) \leftrightarrow ((\SET X = \SET x \cf \psi) \sqcup (\SET X = \SET x \cf \chi))$.

\noindent C3. $(\SET X = \SET x \cf (\alpha\!\supset\!\chi))  \leftrightarrow  ((\SET X = \SET x \cf \alpha) \! \supset \! (\SET X = \SET x \cf \chi))$.

\noindent C4. $(\SET X = \SET x \! \cf \! (\SET Y = \SET y \! \cf \! \chi)) \rightarrow ((\SET X' = \SET x' \land \SET Y = \SET y) \! \cf \! \chi)$ 

\hspace{8pt} (where $\SET X' = \SET X \setminus \SET Y$ and $\SET x' = \SET x \setminus \SET y$; and provided $\SET X = \SET x$

\hspace{8pt} is consistent).



\noindent C4b. $((\SET X = \SET x \land \SET Y = \SET y) \! \cf \! \chi) \rightarrow (\SET X = \SET x \! \cf \! (\SET Y = \SET y \! \cf \! \chi))$.




\noindent C5. $(\SET X = \SET x \cf \bot)\rightarrow \psi$ (when $\SET X = \SET x$ is consistent).

\noindent C6. $(\SET X = \SET x \land Y = y) \cf Y = y$. 


\noindent C7. $(\SET X = \SET x \land \gamma) \rightarrow (\SET X = \SET x \cf \gamma)$

\hspace{14pt} (where $\gamma\in\PCO$ without occurrences of $\cf$).




\noindent C8. $(\SET X = \SET x \cf \Pr(\alpha)\vartriangleright\epsilon) \leftrightarrow \Pr(\SET X = \SET x \cf \alpha)\vartriangleright\epsilon$.  

\hspace{16pt}(where $\vartriangleright=\geq$ or $>$).

\noindent C8b. $(\SET X = \SET x \cf \Pr(\alpha)\vartriangleright \Pr(\beta)) \leftrightarrow \Pr(\SET X = \SET x \cf \alpha)\vartriangleright$ 

\hspace{16pt} $\Pr(\SET X = \SET x \cf \beta)$ \hspace{10pt}(where $\vartriangleright=\geq$ or $>$).

\noindent C9. $\varphi_{End(Y)}\rightarrow (\SET W_Y = \SET w \cf \bigsqcup_{y\in \ran(Y)} Y=y)$.

\noindent C10. $(\varphi_{End(Y)})^C \rightarrow (Y=y \supset (\SET W_V=\SET w \cf Y=y))$ 

\hspace{16pt}(when $\gamma$ has no occurrences of $\cf$).

\noindent C11. 
$(X_1 \leadsto X_2 \land \dots \land X_{n-1} \leadsto X_n)\rightarrow (X_n \leadsto X_1)^C$ 

\hspace{20pt} (for $n>1$).



\vspace{5pt}

\noindent Rule Mon$_\cf$. \LARGE$\frac{\vdash \psi\rightarrow \chi   }{\vdash (\SET X = \SET x \cf \psi) \rightarrow (\SET X = \SET x \cf \chi)}$\normalsize 

\vspace{5pt}

\noindent We will refer to this list of axioms and rules as the 
 \textbf{deduction system}, and write $\Gamma \vdash \varphi$ if there is a countable sequence of $\PCO$ formulas $\varphi_1,\dots,\varphi_\kappa=\varphi$ (enumerated by ordinals $\leq\kappa$) where each formula in the list is either an axiom, a formula from $\Gamma$, or it follows from earlier formulas in the list by one of the rules. The sequence itself is called a \textbf{proof}. 

 We write $\vdash \varphi$ for $\emptyset \vdash \varphi$; if it holds, we say that $\varphi$ is a \textbf{theorem}.  Notice that some of the rules (Rep, Mon$_\supset$, Mon$_\cf$, $\rightarrow$to$\supset$) can only be applied to theorems, since they preserve validity but not truth.













\subsection{Remarks on the axioms and rules}



Rules MP, $\bot^\omega$, $\supset^\omega$ and axioms P1-2-3-4-5 and O1-2-3 are essentially adapted from the paper \cite{RasOgnMar2004}.
Our restriction $\delta+\epsilon\leq 1$ in axiom P3 is imposed by the syntax (we do not allow numbers greater than $1$ as symbols). The additional axiom P3b guarantees that, in fact, axiom scheme P3 is always applicable, in the sense that, if an instance of it is not admitted as an axiom, then the premises of said instance are contradictory.\footnote{It seems to us that an axiom analogous to P3b should be added also to the system in \cite{RasOgnMar2004}.}
Axiom P6 derives from \cite{RasOgnMar2004}, but in our case the correct formulation requires the interaction of the two conditionals $\supset$ (used to define $\equiv$) and $\rightarrow$; notice that $\Pr(\alpha\leftrightarrow\beta)=1$ is not a $\PCO$ formula, and that the analogous formulation $(\alpha\leftrightarrow\beta)\rightarrow (\Pr(\alpha)=\epsilon\rightarrow \Pr(\beta)=\epsilon)$ is \emph{not} valid. The variant P6b is our addition. 
These adaptations are due both to differences in the syntax (\cite{RasOgnMar2004} has an explicit conditional probability operator, while we talk of conditional probabilities only indirectly, by means of the  selective implication; and we have distinct logical operators at the level of events vs. the level of probabilities) and in the semantics (in particular, we differ in the treatment of truth over empty models). 

Our 
 Rule Mon$\supset$  allows omitting the axioms 8, 11 and 12 from \cite{RasOgnMar2004}, which are here proved in Lemma \ref{lemma: 8 11 12}. 
 
 Analogues of CP1-2 appear, for example, in \cite{OgnPerRas2008}, and in earlier literature.

Axioms C6, C7 and C11 take the same roles as the principles of \emph{Effectiveness}, \emph{Composition} and \emph{Recursivity} from \cite{GalPea1998}. 
The current, more intuitive form of axiom C7 was introduced in \cite{BarYan2022}. 
Halpern \cite{Hal2000} noticed that $\cf$ distributes over boolean operators, and formulated analogues of C1 and C2. The validity of C3-4-4b was pointed out in \cite{BarSan2020}, and the importance of C5 in \cite{BarYan2022}.

\section{Soundness}\label{sec: PCO soundness}
As already explained, many of the axioms and rules come from the literature and they are easily seen to be correct also in our semantic framework. The soundness of T1 follows from the fact that it involves essentially classical operators. The soundness of T2 is a consequence of the flatness of $\CO$; this point is explained in \cite{BarSan2020}, Appendix C (in particular, Lemmas C.1 and C.2). We prove in detail the soundness of the particularly complex or unusual axioms C7-9-10. 

\begin{proposition}[Soundness of axiom C7]
Let $\gamma$ be a $\PCO$ formula without  occurrences of $\cf$. Then 
 $T\models(\SET X = \SET x \land \gamma) \rightarrow (\SET X = \SET x \cf \gamma)$ for every $T$ of the appropriate signature.
\end{proposition}

\begin{proof}(Sketch)
Since $T=(T^-,\F)\models \SET X = \SET x$, it is easy to see (by induction on an appropriate notion of distance in the causal graph) that $(T_{\SET X = \SET x})^- = T^-$. Since $\gamma$ has no occurrences of $\cf$ it can be proved (by induction on $\gamma$) that $(T^-,\G)\models \gamma$ iff $T\models \gamma$. In particular, since $(T_{\SET X = \SET x})^- = T^-$ and $T\models \gamma$, we have $T_{\SET X = \SET x} \models \gamma$, and thus $T\models \SET X = \SET x \cf \gamma$. 
\end{proof}

The remainging soundness proofs require some lemmas.

\begin{lemma}\label{lemma: endogenous variables have a direct cause}
 Let $T$ be a causal multiteam. Suppose variable $Y$ is endogenous in $T$. Then there is a variable $X\in \SET W_Y$ and values $x\neq x'\in Ran(X)$, $y\neq y'\in Ran(Y)$ and $\SET w\in Ran(\SET W_{XY})$ such that   \[
T\models(\SET W_{XY}=\SET w \land X=x)\cf Y=y
\]
and
\[
T\models(\SET W_{XY}=\SET w \land X=x')\cf Y=y',
\]
so that $T\models \varphi_{DC(X,Y)}$.
\end{lemma}

\begin{proof}
Since $Y$ is endogenous, it has an associated function $\F_V$ that is not constant. Then, there are at least two distinct interventions on $Y$ that produce distinct values for $Y$; say, up to some reordering there are two tuples of distinct variables $\SET Z\SET X$ such that $\SET Z\SET X = \SET W_Y$ and values $\SET z\in \ran(\SET Z), \SET x\neq \SET x'\in \ran(\SET X)$, $y\neq y'\in \ran(Y)$ such that   
 \[
T\models(\SET Z=\SET z \land \SET X=\SET x)\cf Y=y
\]
and
\[
T\models(\SET Z=\SET z \land \SET X=\SET x')\cf Y=y' \hspace{20pt}  (*)
\]
hold. Let us write more explicitly $\SET X = X_1\dots X_n$, $\SET x=x_1\dots x_n$, $\SET x'=x_1'\dots  x_n'$. For uniformity of notation let us rename the rightmost variable in $\SET Z$ as $X_0$, and its value as $x_0$. Now, there must be an $i\in \{1,\dots,n\}$ such that 
 \[
T\models(\SET Z=\SET z \land X_1\dots X_{i-1}=x_1'\dots x_{i-1}' \land X_i\dots X_n = x_i\dots x_n)
\]

\vspace{-15pt}

\[
\hspace{-132pt} \cf Y=y
\]
while
\[
T\models(\SET Z=\SET z \land X_1\dots X_{i}=x_1'\dots x_{i}' \land X_{i+1}\dots X_n = x_{i+1}\dots x_n)
\]

\vspace{-15pt}

\[
\hspace{-130pt} \cf Y=y^*
\]
for some $y^*\neq y$. Indeed, if it were not so, we would obtain at the step $i=n$ that $(\SET Z=\SET z \land \SET X=\SET x')\cf Y=y$, contradicting (*). But then, the variable $X$ we are looking for is $X_i$, where $i$ is the first index where $y^*$ is obtained.
\end{proof}

\begin{lemma}\label{lemma: DC and End}
Let $T = (T^-,\F)$ be a causal multiteam of signature $\sigma$, and consider the versions of $\varphi_{DC(X,Y)}$ and $\varphi_{End(V)}$ for signature $\sigma$. Then:
\begin{enumerate}
    \item $T \models \varphi_{DC(X,Y)}$ iff $Y$ is endogenous and $X$ is a non-dummy argument of $\F_Y$.
    \item $T \models \varphi_{End(Y)}$ iff  $Y$ is endogenous.
\end{enumerate}
\end{lemma}

\begin{proof}
   1) Suppose $T \models \varphi_{DC(X,Y)}$.  Then, by the definition of $\varphi_{DC(X,Y)}$ there are $x\neq x'\in \ran(X), y\neq y' \in \ran(Y), \SET w\in Ran(\SET W_Y)$ such that 
\[
T\models(\SET W_{XY}=\SET w \land X=x)\cf Y=y
\]
\[
T\models(\SET W_{XY}=\SET w \land X=x')\cf Y=y'.
\]   
Since $y\neq y'$, this means that some values in the $Y$-column of $T$ can be changed by some intervention on variables distinct from $Y$. Thus $Y$ is not an exogenous variable and it has a corresponding function $\F_Y$. Now we must have 
\[
T\models(\SET W_{XY}=\SET w \land X=x)\cf Y=\F_Y(\SET wx)
\]
\[
T\models(\SET W_{XY}=\SET w \land X=x')\cf Y=\F_Y(\SET wx'),
\]   
from which it follows that $\F_Y(\SET wx) = y \neq y' = \F_Y(\SET wx')$. Thus $X$ is not a dummy argument of $Y$. 

Vice versa, if $Y$ is endogenous and $X$ is a non-dummy argument of $\F_Y$, then there are $\SET w\in \ran(\SET W_Y)$ and $x\neq x'\in \ran(X)$ such that $\F_Y(\SET wx) \neq \F_Y(\SET wx')$. Also, by definition of intervention,  
\[
T\models(\SET W_{XY}=\SET w \land X=x)\cf Y=\F_Y(\SET wx)
\]
\[
T\models(\SET W_{XY}=\SET w \land X=x')\cf Y=\F_Y(\SET wx')
\]   
Since $\F_Y(\SET wx) \neq \F_Y(\SET wx')$, these, together, entail that $T \models \varphi_{DC(X,Y)}$.

2) Suppose $T \models \varphi_{End(Y)}$. Then, for some $X\neq Y$, $T \models \varphi_{DC(X,Y)}$. So, by 1., in particular $Y$ is endogenous.

Suppose $Y$ is endogenous. Then, by Lemma \ref{lemma: endogenous variables have a direct cause}, the function $\F_V$ has at least one non-dummy argument $X$. Then, by 1., 
 $T\models \varphi_{DC(X,Y)}$. Thus $T \models \varphi_{End(Y)}$.
\end{proof}

\begin{proposition}[Soundness of axiom C9]
For any $T$ of appropriate signature, 
    $$T\models\varphi_{End(Y)}\rightarrow (\SET W_Y = \SET w \cf \bigsqcup_{y\in \ran(Y)} Y=y).$$
\end{proposition}

\begin{proof} 
Suppose $T=(T^-,\F)\models \varphi_{End(Y)}$. 
Then, by Lemma \ref{lemma: DC and End}, 2.,  $Y$ is an endogenous variable. But then (for any $\SET w \in \SET W_Y$) all assignments $s\in T_{\SET W_Y = \SET w}^-$ are identical (save for their $Key$), i.e. they are all such that $s(\SET W_Y)= \SET w$ and $s(Y) = \F_Y(\SET w)$.
Thus  $T_{\SET W_Y = \SET w}\models Y = \F_Y(\SET w)$. Then $T_{\SET W_Y = \SET w}\models \bigsqcup_{y\in \ran(Y)} Y=y$, and finally $T\models \SET W_Y = \SET w \cf \bigsqcup_{y\in \ran(Y)} Y=y$. 
\end{proof}

\begin{proposition}[Soundness of axiom C10]
For any $T$ of appropriate signature, 
    $$T\models(\varphi_{End(Y)})^C \rightarrow (Y=y \supset (\SET W_V=\SET w \cf Y=y)).$$
\end{proposition}

\begin{proof}
Suppose $T\models(\varphi_{End(Y)})^C$. 
Then, by Lemma \ref{lemma: DC and End}, 2., $Y$ is exogenous. 
Since $Y$ is exogenous, in each assignment the value of $Y$ is not affected by interventions on $\SET W_Y$. Thus, in particular, since $T^{Y=y}\models Y=y$, we have $T^{Y=y}\models \SET W_Y = \SET w\cf Y=y$, and thus $T\models Y=y \supset (\SET W_Y = \SET w\cf Y=y)$. 
\end{proof}

\section{Completeness}\label{sec: PCO completeness}

We establish that the deduction system is \textbf{strongly complete} for $\PCO$, i.e., that for all $\Gamma\cup\{\varphi\} \subseteq \PCO$,
$
\Gamma \models \varphi \Rightarrow \Gamma \vdash \varphi.
$

\subsection{Derived rules}

First, let us justify four derived rules that we will be often using implicitly.

\begin{theorem}[Deduction] \label{thm: deduction theorem}
 If $\Gamma\cup\{\varphi\} \vdash \psi$, then $\Gamma \vdash \varphi \rightarrow \psi$.
\end{theorem}

\begin{proof}
We proceed by transfinite induction on the length of a proof of $\Gamma\cup\{\varphi\} \vdash \psi$.

\begin{itemize}
\item The base cases for $\psi\in\Gamma\cup \{\varphi\}$ or $\psi$ being an axiom are treated as in any introduction to logic. The same goes for the case when $\psi$ is obtained via Rule MP. For the case that $\psi$ is obtained via Rule $\bot^\omega$, see \cite{RasOgnMar2004}, theorem 4.



\item $\psi = \theta\rightarrow (\alpha \supset \Pr(\beta) = \delta)$ is obtained by applying Rule $\supset^\omega$. Thus, for each $\epsilon\in (0,1]\cap \mathbb Q$,
\begin{align*}
& \Gamma \cup \{\varphi\} \vdash \theta\rightarrow (\Pr(\alpha\land\beta)=\delta\epsilon \leftrightarrow \Pr(\alpha)=\epsilon) \\
 & \Gamma  \vdash \varphi \rightarrow (\theta\rightarrow (\Pr(\alpha\land\beta)=\delta\epsilon \leftrightarrow \Pr(\alpha)=\epsilon)) \tag{i.h.}\\
 & \Gamma  \vdash (\varphi \land \theta)\rightarrow (\Pr(\alpha\land\beta)=\delta\epsilon \leftrightarrow \Pr(\alpha)=\epsilon)) \tag{T1+Rule MP}\\
 \end{align*}
 Then,
\begin{align*} 
 & \Gamma  \vdash (\varphi \land \theta)\rightarrow (\alpha \supset \Pr(\beta) = \delta) \tag{Rule $\supset^\omega$}\\
 & \Gamma  \vdash \varphi \rightarrow  (\theta \rightarrow  (\alpha \supset \Pr(\beta) = \delta)). \tag{T1+Rule MP}\\
\end{align*}

\item Case $\psi$ is obtained via rule Rep, Mon$_\cf$, Mon$_\supset$ or $\rightarrow$to$\supset$. In each of these cases, since the rule is only applicable to theorems, it yields a theorem; i.e., $\vdash \psi$. On the other hand, by T1, $\vdash \psi\rightarrow(\varphi\rightarrow\psi)$. Thus, by Rule MP, $\vdash \varphi\rightarrow\psi$. Thus $\Gamma\vdash \varphi\rightarrow\psi$.
\end{itemize}
\end{proof}

\begin{theorem}[Deduction theorem for $\supset$]\label{thm: deduction theorem for supset}
Let $\alpha\in \CO$. If $\alpha\vdash \eta$, then $\vdash \alpha \supset \eta$. If furthermore $\eta\in \CO$, then $\vdash \Pr(\alpha)= \epsilon \rightarrow \Pr(\eta)\geq \epsilon$.
\end{theorem}

\begin{proof}
Suppose $\alpha\vdash \eta$. By the deduction theorem (\ref{thm: deduction theorem}), $\vdash \alpha \rightarrow \eta$. By Rule $\rightarrow$to$\supset$, $\vdash \alpha \supset \eta$. 

Now. if $\eta\in\CO$, by P1 and Rule MP, $\vdash\Pr(\alpha \supset \eta)=1$. By P6b and Rule MP, $\vdash \Pr(\alpha)= \epsilon \rightarrow \Pr(\eta)\geq \epsilon$.    
\end{proof}


\begin{lemma}\label{lemma: P8* and P8**}
Whenever $\alpha,\beta\in\CO$:
\begin{enumerate}
    \item $\vdash \Pr(\alpha\land\beta)=1\rightarrow\Pr(\alpha)=1$.
    \item $\vdash (\Pr(\alpha)=1 \land \Pr(\beta)=1) \rightarrow \Pr(\alpha\land\beta)=1$
\end{enumerate}
\end{lemma}


\begin{proof}
   1) Suppose $\vdash \Pr(\alpha\land\beta)=1$. Then, by P1, $\vdash \alpha \land \beta$. Then, by T1, $\vdash \alpha$. Again by P1, $\vdash \Pr(\alpha)=1$.

   2) Suppose $\vdash \Pr(\alpha)=1 \land \Pr(\beta)=1$. Then, by T1, $\vdash \Pr(\alpha)=1$ and $\vdash \Pr(\beta)=1$. By P1, $\vdash \alpha$ and $\vdash \beta$. Then $\vdash \alpha \land \beta$ by T1. Finally, $\vdash\Pr(\alpha \land \beta)=1$ by P1. 
\end{proof}

\begin{lemma}[Modus ponens for $\supset$]\label{lemma: MP for supset}
Let $\alpha\in\CO_{\sigma}$ and $\Gamma\cup\{\psi\}\subseteq\PCO_{\sigma}$. Then,
\begin{enumerate}
\item If $\Gamma\vdash \alpha\supset\psi$ and $\Gamma\vdash\alpha$, 
then $\Gamma\vdash\psi$.

\item
if $\Gamma\vdash \Pr(\alpha\supset\beta)=1$ and $\Gamma \vdash \Pr(\alpha)=1$, then $\Gamma \vdash \Pr(\beta)=1$.

\end{enumerate}
\end{lemma}

\begin{proof}
1) Assume $\Gamma\vdash\alpha$ and $\Gamma\vdash \alpha\supset\psi$. From the latter we get, by axiom O4 and Rule MP, that $\Gamma\vdash \alpha\rightarrow\psi$. Thus, by Rule MP again, $\Gamma\vdash\psi$.



2) Part 2) can be obtained from part 1) via a few applications of axiom P1.
\end{proof}

\begin{proposition}[Contraposition] \label{prop: contraposition}
Suppose $\Gamma,\chi^C \vdash \psi^C$. Then $\Gamma,\psi \vdash \chi$.
\end{proposition}

\begin{proof}
Assume $\Gamma,\chi^C \vdash \psi^C$. Then by classical logic (axiom T1+Rule MP) $\Gamma,\chi^C, \psi \vdash \bot$. Thus, by the deduction theorem (\ref{thm: deduction theorem}) $\Gamma,\psi \vdash \chi^C \rightarrow \bot$. Thus, again by classical logic, $\Gamma,\psi\vdash \chi$.
\end{proof}



With the help of these additional proof methods, it is not difficult to prove that our proof system is strongly complete for $\CO$ formulas. This result will be used later in order to replace a very long proof-theoretical argument with a brief semantic proof.

\begin{theorem}[$\CO$-completeness]\label{thm: CO completeness}
Let $\sigma$ be a signature and $\Gamma\cup\{\alpha\}\subseteq \CO$. Then $\Gamma\models\alpha$ implies $\Gamma\vdash \alpha$, where the latter is witnessed by a finite derivation.    
\end{theorem}

\begin{proof}
The paper \cite{BarYan2022} proves the strong completeness of a (finitary) system of axioms and (natural deduction style) rules for $\CO$ formulas.\footnote{In that paper, $\lor$ and $\neg$ are part of the official syntax, while $\alpha\supset \beta$ abbreviates $\neg\alpha\lor\beta$.}
We will show that the axioms and rules of said system are derivable in finitely many steps in our system, from which the statement immediately follows.

The axioms \textsf{ValDef} and \textsf{ValUnq} of \cite{BarYan2022} are special cases of our axioms A3 and A1, respectively. 

The rules for $\land,\lor,\neg$ in \cite{BarYan2022} are classical, so they follow easily from instances of T2. More precisely, rules without discharged assumptions (from $\gamma_1,\dots,\gamma_n$ infer $\alpha$) correspond to instances of T2 of the form  $\vdash (\gamma_1\land\dots\land\gamma_n)\supset \alpha$; 
thus, by Lemma \ref{lemma: MP for supset}, $\gamma_1\land\dots\land\gamma_n\vdash\alpha$, and by T2 again $\gamma_1,\dots,\gamma_n\vdash \alpha$. 

Rules for $\lor$ and $\neg$ with discharged assumptions are dealt with by turning the discharged assumptions into antecedents of $\supset$. For example, the elimination rule for $\lor$ (from $\Gamma\vdash \beta\lor\gamma$, $\Gamma,\beta\vdash\alpha$ and $\Gamma,\gamma\vdash\alpha$ conclude $\Gamma\vdash \alpha$) is derived by observing that $\vdash [(\beta\lor\gamma)\land(\beta\supset\alpha)\land(\gamma\supset\alpha)] \supset \alpha$ as an instance of T2; and then proceeding as in the previous case.

Let us turn to the axioms and rules for $\cf$. The axiom $\cf$\textsf{ExFalso}, $(\SET Y = \SET y \land X=x \land X=x') \cf \bot$ (when $x\neq x'$), is derivable in our system as follows. By axiom C6 we have (*): $\vdash(\SET Y = \SET y \land X=x \land X=x') \cf X=x$ and (**): $\vdash(\SET Y = \SET y \land X=x \land X=x') \cf X=x'$. By A1, $\vdash X=x \rightarrow X\neq x'$ and thus $\vdash X=x \rightarrow \neg X = x'$ by axiom A2. By (*) and Mon$_\cf$, then, $\vdash(\SET Y = \SET y \land X=x \land X=x')  \cf \neg X=x'$. Together with (**), by C1, we get $\vdash(\SET Y = \SET y \land X=x \land X=x')  \cf (X=x' \land\neg X=x')$.  By T2 we have that $\vdash (X=x' \land\neg X=x')\supset \bot$, and then by O4 $\vdash (X=x' \land\neg X=x')\rightarrow \bot$. Thus, by by Mon$_\cf$ again we obtain $\vdash(\SET Y = \SET y \land X=x \land X=x')  \cf\bot$.

The rule $\neg\hspace{-4pt}\cf$\textsf{E} (from $\neg (\SET X = \SET x\cf \alpha)$ infer $\SET X = \SET x\cf\neg\alpha$) is derivable as follows. $\neg (\SET X = \SET x\cf \alpha)$ abbreviates $(\SET X = \SET x\cf \alpha)\supset \bot$. By axioms O1b and O1 we obtain $(\SET X = \SET x\cf \alpha)\supset (\SET X = \SET x\cf\bot)$. By C3, we get $\SET X = \SET x\cf (\alpha\supset \bot)$, which can be abbreviated as $\SET X = \SET x\cf \neg\alpha$. 

Rule $\cf$\textsf{RPL}$_A$ corresponds in a straightforward way to a special case of our rule Rep,   and rule $\cf$\textsf{RPL}$_C$ corresponds to Mon$_\cf$. The remaining axioms and rules for $\cf$ correspond to our axioms C1-C4-C4b-C5-C6-C7-C11.
\end{proof}

\subsection{Consequences of the axioms}

Let us prove a few important consequences of the axioms.

\begin{lemma}\label{lemma: reverse axioms}
Let $\alpha\in\CO$ and $\delta,\epsilon\in [0,1]\cap\mathbb Q$.
\begin{enumerate}
\item $\vdash \Pr(\alpha)\geq\delta \rightarrow \Pr(\alpha)> \epsilon$ (if $\epsilon < \delta$).

\item $\vdash \Pr(\alpha)>\epsilon \rightarrow \Pr(\alpha)\geq\epsilon$.
 \end{enumerate}
\end{lemma}

\begin{proof}
1) By contraposition (Proposition \ref{prop: contraposition}) and the deduction theorem (\ref{thm: deduction theorem}) it suffices to prove $(\Pr(\alpha)> \epsilon)^C \rightarrow (\Pr(\alpha)\geq\delta)^C$, i.e. $\Pr(\alpha) \leq \epsilon \rightarrow \Pr(\alpha) < \delta$. But, since $\epsilon < \delta$, this is guaranteed by axiom P4.

2) By contraposition (Proposition \ref{prop: contraposition}) and the deduction theorem (\ref{thm: deduction theorem}) it suffices to prove $(\Pr(\alpha) \geq \epsilon)^C \rightarrow (\Pr(\alpha)>\epsilon)^C$, i.e. $\Pr(\alpha) < \epsilon \rightarrow \Pr(\alpha) \leq \epsilon$. But this is an instance of axiom P5.
\end{proof}

\begin{lemma}\label{lemma: probability facts}
Let $\alpha\in\CO$ and $\delta,\epsilon\in [0,1]\cap\mathbb Q$.
\begin{enumerate}
\item $\vdash \Pr(\alpha)\geq \delta \rightarrow \Pr(\alpha)\geq \epsilon$ (when $\epsilon < \delta$).
\item $\vdash \Pr(\alpha)\leq \delta \rightarrow \Pr(\alpha)\leq \epsilon$ (when $\epsilon > \delta$).
\item $\vdash \Pr(\alpha) = \delta \rightarrow \Pr(\alpha)\neq \epsilon$ (when $\epsilon \neq \delta$).
\end{enumerate}
\end{lemma}

\begin{proof}
1) By Lemma \ref{lemma: reverse axioms}, 1., we get  $\vdash \Pr(\alpha)\geq \delta \rightarrow \Pr(\alpha)> \epsilon$. By Lemma \ref{lemma: reverse axioms}, 2., we get $\vdash\Pr(\alpha)> \epsilon\rightarrow \Pr(\alpha)\geq \epsilon$. We can then combine the two statements by transitivity (guaranteed by axiom T1 + Rule MP).

2) This proof is analogous, using axioms P4-5 instead of Lemma \ref{lemma: reverse axioms}.

3) Remember that $\Pr(\alpha) = \delta$ abbreviates $\Pr(\alpha) \geq \delta \land \Pr(\alpha) \leq \delta$. Thus from $\vdash \Pr(\alpha) = \delta$ we get $\vdash \Pr(\alpha) \geq \delta$, $\vdash \Pr(\alpha) \leq \delta$ by classical logic. Now we have two cases. If $\delta<\epsilon$, by axiom P4 we obtain $\Pr(\alpha) < \epsilon$, from which $\Pr(\alpha) \neq \epsilon$, i.e. $\Pr(\alpha) > \epsilon \sqcup \Pr(\alpha) < \epsilon$, follows by axiom T1. If instead $\delta>\epsilon$, we obtain $\Pr(\alpha) > \epsilon$ by Lemma \ref{lemma: reverse axioms},1., and proceed to the conclusion in the same way.
\end{proof}

The following correspond to axioms 8, 11 and 12 from \cite{RasOgnMar2004}, which are entailed by our Rule Mon$_\supset$.

\begin{lemma}\label{lemma: 8 11 12}
Let $\alpha,\beta\in\CO$ and $\delta,\epsilon\in [0,1]\cap\mathbb Q$.
\begin{enumerate}
\item $\vdash (\alpha\supset \Pr(\beta)=\delta) \rightarrow (\alpha\supset \Pr(\beta)\neq \epsilon)$ when $\delta\neq\epsilon$.
\item $\vdash (\alpha\supset \Pr(\beta)=\delta) \rightarrow (\alpha\supset \Pr(\beta) < \epsilon)$ when $\delta <\epsilon$.
\item $\vdash (\alpha\supset \Pr(\beta)=\delta) \rightarrow (\alpha\supset \Pr(\beta)\geq \epsilon)$ when $\delta \geq\epsilon$.
\end{enumerate}
\end{lemma}

\begin{proof}
1) This is an immediate consequence of Lemma \ref{lemma: probability facts},3., and Rule Mon$_\supset$.

2) If $\delta <\epsilon$, then $\vdash\Pr(\beta)\leq\delta \rightarrow \Pr(\beta) < \epsilon$ by axiom P4. Since $\Pr(\beta)=\delta$ abbreviates $\Pr(\beta)\geq\delta \land \Pr(\beta)\leq\delta$, we have $\vdash\Pr(\beta) = \delta \rightarrow \Pr(\beta) < \epsilon$ by classical logic (T1 + Rule MP). The thesis then follows by Rule Mon$_\supset$.

3) If $\delta = \epsilon$, then by classical logic $\Pr(\beta)\geq\delta \vdash \Pr(\beta)\geq\epsilon$, and $\Pr(\beta)=\delta \vdash\Pr(\beta)\geq\delta$. Thus $\Pr(\beta)=\delta \vdash\Pr(\beta)\geq\epsilon$.

If instead $\delta > \epsilon$, by Lemma \ref{lemma: probability facts}, 1., we obtain $\Pr(\beta)\geq\delta \vdash\Pr(\beta)\geq\epsilon$. By classical logic, then, $\Pr(\beta)=\delta \vdash\Pr(\beta)\geq\epsilon$. 

Finally, in both cases, we obtain the conclusion by Rule Mon$_\supset$.
\end{proof}

\begin{lemma}\label{lemma: axiom P8}
For all $\alpha,\beta\in\CO$, 
\[
\vdash(\Pr(\alpha\lor\beta)=\epsilon \land \Pr(\beta)=0)\rightarrow \Pr(\alpha)=\epsilon.
\]  
\end{lemma}

\begin{proof}
Write $\varphi$ for the antecedent $\Pr(\alpha\lor\beta)=\epsilon \land \Pr(\beta)=0$. we have:
\begin{align*}
1. \ \ \varphi & \vdash \Pr(\beta)=0 \rightarrow (\beta \supset \bot) \tag{Axiom O1}\\   
2. \ \ \varphi & \vdash \beta \supset \bot \tag{1., Rule MP}\\    
3. \ \ \varphi & \vdash \Pr(\beta \supset\bot)=1 \tag{2., P1 + Rule MP}\\
4. \ \ \varphi & \vdash \Pr([(\alpha\lor\beta) \land (\beta\supset \bot)]\equiv \alpha) = 1 \tag{Axiom T2+P1}\\    
5. \ \ \varphi & \vdash \Pr\Big((\beta\supset \bot) \supset \{[(\alpha\lor\beta) \land (\beta\supset \bot) \equiv \alpha] \} \\
 & \hspace{15pt} \supset \{ \alpha \lor \beta \equiv \alpha \}\Big) = 1 \tag{Axiom T2+P1}\\  
6. \ \ \varphi & \vdash \Pr(\alpha \lor \beta \equiv \alpha)=1 \tag{3.,4.,5.,  $2\times$Lemma \ref{lemma: MP for supset}}\\
7. \ \ \varphi & \vdash \Pr(\alpha \lor \beta) = \epsilon \rightarrow \Pr(\alpha)=\epsilon \tag{P6+Rule MP}\\ 
8.  \ \ \varphi & \vdash \Pr(\alpha)=\epsilon\tag{7., Rule MP}\\ 
9.  \ \ \phantom{\varphi} & \vdash \varphi \rightarrow  \Pr(\alpha)=\epsilon. \tag{8., Theorem \ref{thm: deduction theorem}}
\end{align*}
\end{proof}

\subsection{Maximal consistent sets of formulas}

We now move towards a Lindenbaum lemma. We say that a set $\Gamma$ of $\PCO$ formulas of a fixed signature $\sigma$ is \textbf{consistent} if $\bot$ is not derivable from it (or, equivalently, there is at least one formula that is not derivable from it). $\Gamma$ is maximally consistent if it is consistent and, if $\Gamma'\supseteq \Gamma$ and $\Gamma'$ is consistent, then $\Gamma' = \Gamma$.

\begin{lemma}[Lindenbaum]\label{lemma: Lindenbaum}
Every consistent set $\Gamma$ of 
$\PCO$ formulas of a fixed signature $\sigma$  can be extended to a maximally consistent set $\Delta$, which furthermore is closed under derivations (i.e. $\Delta\vdash \psi$ implies $\psi\in\Delta$).
\end{lemma}

\begin{proof}
Let $\Gamma$ be a consistent set, and $\Gamma^\vdash \supset \Gamma$ the set of its formal consequences. We enumerate the formulas of $\PCO_{\sigma}$ as $\varphi_1,\dots,\varphi_n,\dots$ 
 and the formulas of $\CO_{\sigma}$ as $\alpha_1,\dots,\alpha_n,\dots$ using natural numbers as indexes.

We then define an increasing sequence of sets of $\PCO_{\sigma}$ formulas $\Gamma_i$ as follows:
\begin{enumerate}
\item $\Gamma_0 := \Gamma^\vdash$.
\item $\Gamma_{2i+1}$ is defined as follows:
	\begin{itemize}
	\item If $\Gamma_{2i}\cup\{\varphi_i\}$ is consistent, let $\Gamma_{2i+1}:=\Gamma_{2i}\cup\{\varphi_i\}$.  
	\item Otherwise, if $\varphi_i$ is of the form $\psi\rightarrow(\alpha\supset \Pr(\beta)=\epsilon)$, let $\Gamma_{2i+1}:=\Gamma_{2i} \cup \{\varphi_i^C, \psi\rightarrow (\Pr(\alpha\land\beta) = \delta\epsilon \leftrightarrow \Pr(\alpha)=\delta )^C \}$ for some $\delta>0$ that makes this set consistent.
	\item Otherwise, let $\Gamma_{2i+1}:=\Gamma_{2i} \cup \{\varphi_i^C\}$.
	\end{itemize}
\item 
Let $\Gamma_{2i+2} := \Gamma_{2i+1} \cup \{\Pr(\alpha_i) = \epsilon\}$ for some $\epsilon$ that makes this set consistent.
\end{enumerate}

We first prove, by induction below 
$\omega$, that all the $\Gamma_i$ are well-defined and consistent.
\begin{itemize}
\item $\Gamma_0$ is consistent since it is the set of consequences of a consistent set.

\item Notice that either $\Gamma_{2i} \cup \{\varphi_i\}$ or $\Gamma_{2i} \cup \{\varphi_i^C\}$ is consistent (otherwise by the deduction theorem we obtain $\Gamma_{2i}\vdash \varphi_i \leftrightarrow \varphi_i^C$, and thus by classical logic $\Gamma_{2i}$ is inconsistent). In case $\Gamma_{2i} \cup \{\varphi_i^C\}$ is consistent and $\varphi_i$ is $\psi\rightarrow(\alpha\supset \Pr(\beta)=\epsilon)$, we also need to prove that there is a $\delta>0$ such that $\Gamma_{2i} \cup \{\varphi_i^C, \psi\rightarrow (\Pr(\alpha\land\beta) = \delta\epsilon \leftrightarrow \Pr(\alpha)=\delta )^C \}$ is consistent. Suppose this is not the case, Then, for all $\delta\in (0,1]\cap\mathbb Q$, 
\begin{align*}
& \Gamma_{2i}, \varphi_i^C, \psi\rightarrow (\Pr(\alpha\land\beta) = \delta\epsilon \leftrightarrow \Pr(\alpha)=\delta )^C   \vdash  \bot \\
& \Gamma_{2i}, \varphi_i^C   \vdash (\psi\rightarrow (\Pr(\alpha\land\beta) = \delta\epsilon \leftrightarrow \Pr(\alpha)=\delta )^C)^C  \tag{deduction theorem + T1}\\
& \Gamma_{2i}, \varphi_i^C   \vdash \psi\rightarrow (\Pr(\alpha\land\beta) = \delta\epsilon \leftrightarrow \Pr(\alpha)=\delta ).  \tag{by T1: tautology $\neg(A\rightarrow \neg B)\rightarrow(A\rightarrow B)$}\\
\end{align*}
Thus, by Rule $\supset^\omega$, $\Gamma_{2i}, \varphi_i^C   \vdash \psi\rightarrow(\alpha\supset \Pr(\beta)=\epsilon)$, i.e. $\Gamma_{2i}, \varphi_i^C   \vdash \varphi_i$. Thus, by the deduction theorem, $\Gamma_{2i}   \vdash \varphi_i^C \rightarrow \varphi_i$, and finally by a tautology $\Gamma_{2i}   \vdash \varphi_i$, which contradicts the consistency of $\Gamma_{2i} \cup \{\varphi_i^C\}$.  

\item Suppose for the sake of contradiction that (for some $i<\omega$) there is no $\epsilon\in [0,1]\cap\mathbb Q$ such that $\Gamma_{2i+1}\cup\{\Pr(\alpha_i)=\epsilon\}$ is consistent. Now write $\Xi:= \Gamma_{2i+1}\setminus \Gamma_0$ (notice that $\Xi$ is a finite set.) Our assumption gives, for every $\epsilon\in [0,1]\cap\mathbb Q$, that $\Gamma_0,\Xi\vdash \Pr(\alpha_i)\neq \epsilon$ (since $\Pr(\alpha_i)\neq \epsilon$ is $(\Pr(\alpha_i)= \epsilon)^C$). Thus, by the deduction theorem and T1, $\Gamma_0\vdash (\bigwedge_{\theta\in\Xi} \theta)\rightarrow \Pr(\alpha_i)\neq \epsilon$. Thus, by Rule $\bot^\omega$, $\Gamma_0\vdash (\bigwedge_{\theta\in\Xi} \theta)\rightarrow \bot$. Then, by T1 and Rule MP, $\Gamma_0,\Xi\vdash \bot$, contradicting the inductive assumption that $\Gamma_{2i+1}$ is consistent.
\end{itemize}

Let $\Delta := \bigcup_{n\in\mathbb N} \Gamma_n$. We wish to prove that it is consistent. Since either $\varphi_i$ or $(\varphi_i)^C$ is added in $\Gamma_{2i+1}$, this will immediately entail that $\Delta$ is maximally consistent. 
 
 Notice that it cannot be the case that $\bot\in\Delta$, because otherwise we already have $\bot\in\Gamma_i$ for some $i$, thus contradicting the fact that $\Gamma_i$ is consistent.
Thus, if we prove that $\Delta\vdash\psi$ implies $\psi\in\Delta$, we can conclude that $\Delta\not\vdash\bot$, i.e. $\Delta$ is consistent.

We prove it by transfinite induction on the length of a proof $\Delta\vdash\psi$. First, notice that if the last rule applied is finitary (where by inductive assumption all premises are in $\Delta$) then there is an $i$ such that all premises are in $\Gamma_i$ (thus $\Gamma_i\vdash\psi$), and furthermore either $\psi\in\Gamma_i$ or $\psi^C\in\Gamma_i$. 
Since $\Gamma_i$ is consistent, it must be $\psi\in\Gamma_i$. Thus $\psi\in\Delta$. We next consider what happens if $\psi$ is obtained via an infinitary rule.

Suppose $\psi$ is obtained via Rule $\bot^\omega$. Thus, $\psi$ is $\chi\rightarrow\bot$, and the premises (from $\Delta$) are all the formulas  $\theta^\epsilon: \chi\rightarrow \Pr(\alpha)\neq \epsilon$ for $\epsilon \in [0,1]\cap\mathbb Q$ and for some specific $\alpha\in\CO$. By the inductive hypothesis $\theta^\epsilon\in\Delta$. Now by step 3 of the construction there must be an $\epsilon$ and an index $i$ such that $\Pr(\alpha)=\epsilon\in\Gamma_i$. Since $\chi\rightarrow \Pr(\alpha)=\epsilon$ can be derived from $\Pr(\alpha)=\epsilon$ by T1 and Rule MP (so, without using infinitary rules) we can conclude as above that $\chi\rightarrow \Pr(\alpha)=\epsilon \in\Delta$. But since $\chi\rightarrow \bot$ is derivable from this formula and $\theta^\epsilon$ via T1 and Rule MP, we conclude in the same way that $\chi\rightarrow \bot \in \Delta$, as needed. 

Suppose $\psi$ is obtained via Rule $\supset^\omega$. Then $\psi$ is of the form $\chi \rightarrow (\alpha \supset \Pr(\beta) = \epsilon)$, and the premises are all the formulas $\theta^\delta:\chi\rightarrow (\Pr(\alpha\land\beta)=\delta\epsilon \leftrightarrow \Pr(\alpha)=\delta)$ for $\delta\in (0,1]\cap\mathbb Q$, which by the inductive assumption are in $\Delta$. If we assume for the sake of contradiction that $\chi \rightarrow (\alpha \supset \Pr(\beta) = \epsilon)\notin \Delta$, then its contradictory negation $\chi \land (\alpha \supset \Pr(\beta) = \epsilon)^C$ is in $\Delta$ by maximality; and, by step 2 of the construction, there are an $i\in\mathbb N$ and a $\delta\in (0,1]\cap\mathbb Q$ such that $\chi\rightarrow (\Pr(\alpha\land\beta) = \delta\epsilon \leftrightarrow \Pr(\alpha)=\delta )^C\in\Gamma_i$. Since from this and $\theta^\delta$ one derives via T1 and Rule MP that $\chi\rightarrow \bot$, then by consistency of the stages of the construction there must be a $j$ such that $\chi\rightarrow\bot \in \Gamma_j$ and $\chi \land (\alpha \supset \Pr(\beta) = \delta)^C \in \Gamma_j$. But then, by the same line of reasoning, there is a $k$ such that $\bot\in\Gamma_k$, contradicting the consistency of $\Gamma_k$.
\end{proof}

We then prove some properties of maximal consistent sets of formulas.

\begin{lemma}\label{lemma: properties of MCS}
Let $\Gamma$ be a maximal consistent set of formulas. Then the following hold.

\begin{enumerate}
\item $\Gamma$ contains all theorems.
\item $\psi\land\chi\in\Gamma$  iff $\psi\in \Gamma$ and $\chi\in\Gamma$.
\item There is exactly one $\epsilon\in[0,1]\cap \mathbb Q$ such that $\Pr(\alpha)=\epsilon\in\Gamma$.
\item If $\Pr(\alpha)\geq \delta\in\Gamma$, then there is an $\epsilon\geq\delta$ such that $\Pr(\alpha)=\epsilon\in\Gamma$.
\item If $\Pr(\alpha)\leq \delta\in\Gamma$, then there is an $\epsilon\leq\delta$ such that $\Pr(\alpha)=\epsilon\in\Gamma$.
\item If $\Pr(\alpha)>0\in \Gamma$, then there is a unique $\epsilon\in [0,1]\cap \mathbb Q$ such that $\alpha\supset \Pr(\beta)=\epsilon\in\Gamma$.
\item If $\alpha\supset \Pr(\beta)\geq \delta\in\Gamma$, then there is an $\epsilon\geq\delta$ such that $\alpha\supset \Pr(\beta)=\epsilon\in\Gamma$.
\end{enumerate}
\end{lemma}

\begin{proof}

The proofs of 1. and 2. are routine.

3) Unicity: Lemma \ref{lemma: probability facts}, 3. with Rule MP entail that, if $\Pr(\alpha)=\epsilon\in\Gamma$, then $\Pr(\alpha)\neq\delta\in\Gamma$ whenever $\delta\neq\epsilon$. Since $\Gamma$ is consistent, then, $\Pr(\alpha)=\delta\notin\Gamma$.

Existence: suppose that, for all $\epsilon\in [0,1]\cap \mathbb Q$, $\Pr(\alpha)=\epsilon\notin\Gamma$. Then, by maximality of $\Gamma$, 
$\Pr(\alpha)\neq\epsilon\in\Gamma$; and by classical logic $\top \rightarrow \Pr(\alpha)\neq\epsilon\in\Gamma$. Thus, by Rule $\bot^\omega$, $\top\rightarrow\bot\in \Gamma$. Thus again by classical logic, $\bot\in\Gamma$, contradicting the consistency of $\Gamma$.

4) By 3. there is an $\epsilon$ such that $\Pr(\alpha)=\epsilon\in\Gamma$, from which $\Pr(\alpha)\leq\epsilon\in\Gamma$ follows by classical logic. Suppose for the sake of contradiction that $\epsilon < \delta$. Then by axiom P4 the latter implies $\Pr(\alpha)<\delta\in\Gamma$. Since by assumption also $\Pr(\alpha)\geq\delta\in\Gamma$, we contradict the consistency of $\Gamma$.

5) This is analogous to the proof of 4., using Lemma \ref{lemma: reverse axioms},1., instead of axiom P4.

6) Unicity: Suppose $\alpha\supset \Pr(\beta)=\delta\in\Gamma$ for some $\delta\neq\epsilon$.  By Lemma \ref{lemma: 8 11 12}, 1., we also have $\alpha\supset \Pr(\beta)\neq\delta\in\Gamma$, and thus, by axiom O5$_\land$,    $\alpha\supset (\Pr(\beta)=\delta \land \Pr(\beta)\neq\delta)\in\Gamma$. Since $(\Pr(\beta)=\delta \land \Pr(\beta)\neq\delta)\leftrightarrow\bot \in \Gamma$ by T1, we have $\alpha\supset \bot\in \Gamma$ by Rule Rep. By O1b we obtain $\Pr(\alpha) = 0\in \Gamma$, and since we assumed  $\Pr(\alpha) > 0\in \Gamma$, we contradict the consistency of $\Gamma$.

Existence: By 3., there is a unique $\delta$ such that $\Pr(\alpha)=\delta$, and a unique $\delta'$ such that $\Pr(\alpha \land \beta)=\delta'$. Now $\delta$ must be $>0$ (otherwise $\Gamma$ is inconsistent); 
 but then, by axiom O2, $\alpha\supset \Pr(\beta)=\frac{\delta'}{\delta}\in \Gamma$, as needed. 

7) By 6., either $\Pr(\alpha)>0\notin \Gamma$ or there is an $\epsilon$ such that $\alpha\supset \Pr(\beta)=\epsilon\in\Gamma$. In the former case, by the maximality of $\Gamma$ and axiom P2 we obtain $\Pr(\alpha)=0\in \Gamma$. Then by axiom O1 and modus ponens we obtain $\alpha\supset \Pr(\beta)=\delta$, as needed.

In the latter case, we can assume without loss of generality that $\Pr(\alpha)>0\in \Gamma$. Suppose for the sake of contradiction that $\epsilon<\delta$. Then by Lemma \ref{lemma: 8 11 12}, 2., we have $\alpha\supset \Pr(\beta)<\delta\in \Gamma$. Thus by classical logic $\Pr(\alpha)>0 \land \alpha\supset \Pr(\beta)<\delta\in \Gamma$, that is, $(\alpha\supset \Pr(\beta)\geq\delta)^C\in\Gamma$, contradicting the consistency of $\Gamma$.
\end{proof}

We formulate some analogous closure properties for formulas involving counterfactuals. 

\begin{lemma}\label{lemma: cf properties of MCS}
Let $\Gamma$ be a maximal consistent set of formulas. 
 Then the following hold.
\begin{enumerate}
\item If $\varphi_{End(Y)}\in\Gamma$, then for all $\SET w\in \ran(\SET W_Y)$ there is a unique $y\in \ran(Y)$ such that $``\SET W_Y = \SET w \cf Y=y''\in\Gamma$.
\item There is a unique $\epsilon\in [0,1]\cap\mathbb Q$ such that $``\SET X = \SET x \cf \Pr(\alpha)=\epsilon''\in \Gamma$.
\end{enumerate}
\end{lemma}

\begin{proof}
1) Unicity: Suppose $\SET W_Y = \SET w \cf Y=y\in\Gamma$ and $\SET W_Y = \SET w \cf Y=y' \in\Gamma$ for some $y\neq y'$. From the latter, by axiom A1 and Rule Mon$_\cf$ we get $\SET W_Y = \SET w \cf Y\neq y\in\Gamma$. By axiom A2 and Rule Mon$_\cf$ we get $\SET W_Y = \SET w \cf \neg(Y =  y)\in\Gamma$. By axiom O4, T1 and Rule Mon$_\cf$ we get $\SET W_Y = \SET w \cf (Y =  y)^C \in\Gamma$. By axiom C1, then, we have $\SET W_Y = \SET w \cf (Y=y \land (Y =  y)^C)\in\Gamma$. Since by T1 we have $\vdash (Y=y\land (Y =  y)^C) \rightarrow \bot$, by Rule Mon$_\cf$ we get $\SET W_Y = \SET w \cf \bot \in\Gamma$. Finally by C5 we have $\bot\in\Gamma$, a contradiction.

 Existence: 
 if no $y$ satisfies the statement, we have $(\SET W_Y = \SET w \cf Y=y)^C \in \Gamma$ for all $y\in \ran(Y)$, i.e. $\SET W_Y = \SET w \cf \Pr(Y=y)<1 \in \Gamma$, and then, by T1 and Rule MP, $\bigwedge_{y\in \ran(Y)}\SET W_Y = \SET w \cf \Pr(Y=y)<1 \in \Gamma$. On the other hand, since $\varphi_{End(Y)}\in\Gamma$, by axiom C9, $\SET W_Y = \SET w \cf \bigsqcup_{y\in \ran(Y)} Y=y\in\Gamma$, and, by axiom C2, $\bigsqcup_{y\in \ran(Y)}(\SET W_Y = \SET w \cf  Y=y)\in \Gamma$. But since this is $(\bigwedge_{y\in \ran(Y)}\SET W_Y = \SET w \cf \Pr(Y=y)<1)^C$, we contradict the consistency of $\Gamma$.
 
2) Unicity: suppose $\SET X = \SET x \cf \Pr(\alpha)=\epsilon\in\Gamma$ and $\SET X = \SET x \cf \Pr(\alpha)=\epsilon' \in\Gamma$ for some $\epsilon\neq \epsilon'$. Since, by Lemma \ref{lemma: probability facts}, 3., $\Pr(\alpha)=\epsilon' \vdash \Pr(\alpha)\neq\epsilon$, by Rule Mon$_\cf$ we get $\SET X = \SET x \cf \Pr(\alpha)\neq\epsilon\in\Gamma$. Since this is the contradictory negation of $\SET X = \SET x \cf \Pr(\alpha)=\epsilon\in\Gamma$, we find a contradiction with the consistency of $\Gamma$.

Existence: suppose there is no such $\epsilon$. Thus, by maximality of $\Gamma$, for all $\epsilon\in[0,1]\cap\mathbb Q$ we have $\SET X = \SET x \cf \Pr(\alpha)\neq\epsilon \in \Gamma$. By axioms C8 and T1 plus Rule MP, then, $\Pr(\SET X = \SET x \cf \alpha)\neq\epsilon \in \Gamma$. By T1 and Rule MP, $\top\rightarrow\Pr(\SET X = \SET x \cf \alpha)\neq\epsilon \in \Gamma$. Thus, by Rule $\bot^\omega$, $\top\rightarrow\bot \in\Gamma$, and, by T1+Rule MP, $\bot\in\Gamma$, a contradiction.
\end{proof}

Remember that $\SET W$ denotes an injective listing of all the variables in $\dom$. We can assume that all assignments over $\dom$ are enumerated as $s_1,\dots,s_n$. For each $i=1,\dots,n$, we will denote as $\Al_i$ the formula $\SET W = s_i(\SET W)$.


\vspace{15pt}

\begin{lemma}\label{lemma: incompatibility of hat alphas} 
 \phantom{a}
\begin{enumerate}
    \item If $i\neq j$, $\vdash (\Al_i \land \Al_j) \supset \bot$.
    \item If $j\notin I$, 
    $\vdash \Pr\Big((\bigvee_{i\in I}\Al_i)\land\Al_j\Big) = 0$.
\end{enumerate}
    
\end{lemma}

\begin{proof}
1) Since $s_i$ and $s_j$ are distinct, there is a variable $V$ and values $v\neq v'$ such that the conjunct $V=v$ occurs in $\Al_i$ and the conjunct $V=v'$ occurs in $\Al_j$. Now we have:
\begin{align*}
  & 1. \vdash (\Al_i\land\Al_j) \rightarrow (V=v \land V=v') \tag{T1}\\
  & 2. \vdash (V=v \land V=v') \rightarrow \bot \tag{A1 + T1 + Rule MP}\\
  & 3. \vdash (\Al_i\land\Al_j) \rightarrow \bot \tag{1., 2., T1 + Rule MP}\\
  & 4. \vdash  (\Al_i \land \Al_j) \supset \bot. \tag{4., Rule $\rightarrow$to$\supset$}
\end{align*}

2) From point 1. and T2 we obtain that $\vdash ((\bigvee_{i\in I}\Al_i)\land\Al_j) \supset \bot$. The conclusion follows from axiom O1b.
\end{proof}

\begin{lemma}\label{lemma: decomposed probabilities}
Let $\Gamma$ be a maximal consistent set of $\PCO$ formulas of signature $\sigma$, and $I\subseteq \{1,\dots,n\}$, where $n$ is the number of distinct assignments of signature $\sigma$. For each $i\in I$, write $\epsilon_i$ for the unique rational number in $[0,1]$ such that $\Pr(\Al_i)=\epsilon_i\in \Gamma$ (as given by Lemma \ref{lemma: properties of MCS}, 3.). Then, $$\Pr(\bigvee_{i\in I}\Al_i)= \sum_{i\in I}\epsilon_i\in\Gamma.$$ 
\end{lemma}

\begin{proof}
By Lemma \ref{lemma: incompatibility of hat alphas} we have that, if $k\notin J \subseteq I$, then $\vdash \Pr((\bigvee_{i\in J}\Al_i)\land \Al_{k})=0$.



Then, using the family of formulas just proved and the fact that 
 $\Pr(\Al_i)=\epsilon_i\in \Gamma$, for $i=1,\dots,n$, by repeated applications of axiom P3 we obtain $\Pr(\bigvee_{i\in I}\Al_i) = \sum_{i\in I}\epsilon_i\in\Gamma$. We must justify that each of these applications of P3 is permitted, i.e., for a fixed enumeration $i_1,\dots,i_l$ of $I$ and for each $m=1,\dots,l-1$ we must also show that that the sum of $\sum_{k=1\dots m}\epsilon_{i_k}$ and $\epsilon_{i_{m+1}}$ is $\leq 1$. 

We prove it by induction on $m$.  By the inductive assumption we have $\sum_{k=1\dots m}\epsilon_{i_k}\leq 1$ and $\Pr(\bigvee_{k = 1\dots m}\Al_{i_k}) = \sum_{k = 1\dots m}\epsilon_{i_k}\in\Gamma$. By Lemma \ref{lemma: incompatibility of hat alphas} we also have $\Pr((\bigvee_{k=1\dots m}\Al_i) \land \Al_{i_{m+1}}) = 0\in\Gamma$. Thus, by axiom P3b, $\Pr(\Al_{i_{m+1}})\leq 1- \sum_{k = 1\dots m}\epsilon_{i_k}\in\Gamma$. Thus, by lemma \ref{lemma: properties of MCS}, 5.,  there is a $\delta\leq 1- \sum_{k = 1\dots m}\epsilon_{i_k}$ such that $\Pr(\Al_{i_{m+1}})= \delta\in\Gamma$. 
Since furthermore $\Pr(\Al_{i_{m+1}})=\epsilon_{i_{m+1}}\in\Gamma$, by lemma \ref{lemma: properties of MCS}, 3. (unicity) we obtain $\delta = \epsilon_{i_{m+1}}$, and so $\sum_{k=1\dots m}\epsilon_{i_k} + \epsilon_{i_{m+1}} = \sum_{k=1\dots m}\epsilon_{i_k} + \delta \leq \sum_{k=1\dots m}\epsilon_{i_k} + 1- \sum_{k = 1\dots m}\epsilon_{i_k} =1$, as needed.
\end{proof}

\subsection{Canonical causal multiteams}

Let us now see how to build, out of a maximal consistent set of formulas $\Gamma$, a ``canonical'' causal multiteam $\T = (\T^-,\F)$. 
Let again $\SET W$ be an injective listing of all the variables in $\dom$, and let us assume that the assignments over $\dom$ are enumerated as $s_1,\dots,s_n$. We will denote as $\Al_i$ the formula $\SET W = s_i(\SET W)$. Now, by Lemma \ref{lemma: properties of MCS}, 3., for each $i$ there is exactly one $\epsilon_i\in[0,1]\cap \mathbb Q$ such that $\Pr(\hat\alpha_i)=\epsilon_i \in\Gamma$. Let $d$ be a common denominator of the $\epsilon_i$, so that $\epsilon_i = \frac{m_i}{d}$.

\begin{itemize}
\item Let $\T^-$ contain $m_i$ copies of the assignment $s_i$.
\item If $\varphi_{End(Y)}\notin\Gamma$, let $Y$ be an exogenous variable.
\item If $\varphi_{End(Y)}\in\Gamma$, let $Y$ be an endogenous variable, and let $\F_Y: \ran(\SET W_Y)\rightarrow \ran(Y)$ be the function that assigns to $\SET w\in \ran(\SET W_Y)$ the unique value $y\in \ran(Y)$ such that $\SET W_Y = \SET w \cf Y=y$ (whose existence is guaranteed by Lemma \ref{lemma: cf properties of MCS}, 1.)
\end{itemize}

The graph $G_\T$ of $\T$ is then obtained, as usual, by starting from the graph in which each endogenous variable has arrows incoming from all other variables; and then removing all the arrows that correspond to dummy arguments.

We prove some key properties of $\T$.

\begin{lemma}\label{lemma: properties of canonical T}
Let $\T$ be constructed from a maximal consistent set $\Gamma$, as above.
\begin{enumerate}
\item $\T$ is a causal multiteam.
\item $\T$ is recursive, i.e. $G_\T$ is acyclic.
\item $\sum_{i=1..n} \epsilon_i = 1$.
\item $|\T^-|= d$.
\item $P_\T(\hat \alpha_i)=\epsilon_i$.
\item For all $\beta \in \CO_{\sigma}$, $P_\T(\beta)=\sum_{i \mid (s_i,\F)\models \beta} \epsilon_i$.
\end{enumerate}
\end{lemma}

\begin{proof}
1) We need to verify that each assignment that has at least one copy in $\T^-$ is compatible with $\F$. 
If $V$ is an exogenous variable of $\F$ we do not need to verify anything. So let $V$ be endogenous. Suppose the compatibility constraint is not satisfied by a given assignment $s\in \T^-$: then $s(V)\neq\F_V(s(\SET W_V))$. By construction of $\T^-$,
\[
\Pr\big(\SET W_V = s(\SET W_V) \land V = s(V)\big)=\epsilon\in\Gamma, \text{ for some $\epsilon>0$.}
\]
On the other hand, by A1 and T1 we have
\[
\SET W_V = s(\SET W_V) \land V = s(V) \vdash \SET W_V = s(\SET W_V) \land V \neq \F_V(s(\SET W_V)).
\]
Thus, by Theorem \ref{thm: deduction theorem for supset} we obtain
\[
(*):  \Pr\big(\SET W_V = s(\SET W_V) \land V = \F_V(s(\SET W_V))\big)\geq\epsilon\in\Gamma.
\]


On the other hand, by the  construction of $\F$, we obtain that
$\SET W_V = s(\SET W_V) \cf V = \F_V(s(\SET W_V))\in \Gamma$.
By axiom schema T2, $V = \F_V(s(\SET W_V)) \equiv \big(\neg(V = \F_V(s(\SET W_V)))\supset \bot\big)\in \Gamma$, and by O4 and T1, 
\[
V = \F_V(s(\SET W_V)) \leftrightarrow (\neg(V = \F_V(s(\SET W_V)))\supset \bot)\in \Gamma.
\]
Thus $\SET W_V = s(\SET W_V) \cf  (\neg(V = \F_V(s(\SET W_V)))\supset \bot)\in \Gamma$ follow by Rule Rep. By axiom O3,
\begin{multline*}
(\SET W_V = s(\SET W_V) \cf  \neg(V = \F_V(s(\SET W_V))))\\
\supset (\SET W_V = s(\SET W_V) \cf\bot)\in \Gamma.
\end{multline*}
The consequent of the above entails $\bot$ by C5, so by Mon$_\supset$ we have $(\SET W_V = s(\SET W_V) \cf  \neg(V = \F_V(s(\SET W_V))))\supset \bot\in \Gamma$, i.e. $\neg(\SET W_V = s(\SET W_V) \cf  \neg(V = \F_V(s(\SET W_V))))\in \Gamma$. By A2 and Rep we get $\neg(\SET W_V = s(\SET W_V) \cf  V \neq  \F_V(s(\SET W_V)))\in \Gamma$. Then, by P1, 
$\Pr\big(\neg\big(\SET W_V = s(\SET W_V) \cf V \neq \F_V(s(\SET W_V))\big)\big)=1\in \Gamma$. By classical logic 
\[
\Pr\big(\neg(\SET W_V = s(\SET W_V) \cf V \neq \F_V(s(\SET W_V)))\big)\geq 1\in \Gamma,
\]
and thus, by definition of the abbreviation for $\leq$, 
\[
\Pr(\SET W_V = s(\SET W_V) \cf V \neq \F_V(s(\SET W_V)))\leq 0\in \Gamma.
\]
Thus, by P2, $\Pr(\SET W_V = s(\SET W_V) \cf V \neq \F_V(s(\SET W_V)))= 0\in \Gamma$. Thus, by  
 P4, $\Pr(\SET W_V = s(\SET W_V) \cf V \neq \F_V(s(\SET W_V))) < \epsilon\in \Gamma$, which contradicts (*).

2) First we prove that $X\in \PA_Y$ implies that $\varphi_{DC(X,Y)}\in \Gamma$. Suppose $\varphi_{DC(X,Y)}\notin \Gamma$. Then, since $\Gamma$ is closed under proofs, by axiom T1 it must also be the case that (**):
\[
\hspace{-80pt}  ((\SET W_{XY}=\SET w \land X=x)\cf Y=y) \land 
\]
\[
((\SET W_{XY}=\SET w  \land X=x')\cf Y=y')\notin\Gamma
\]
for every $\SET w$, $x\neq x'$, $y\neq y'$.

Notice that we can assume $Y$ is an endogenous variable; otherwise, automatically $X\notin \PA_Y = \emptyset$. Remember that then, by the construction of $\T$, we have $\varphi_{End(Y)}\in \Gamma$.

Since $Y$ is endogenous, by definition of $\T$, 
for all $\SET wx$ we have $(\SET W_{XY} = \SET w \land X =  x)\cf Y=\F_Y(\SET wx)\in\Gamma$. By the appropriate instances of (**), then, we must conclude that, for all $x'\neq x$ and $y'\neq \F_Y(\SET wx)$, $(\SET W_{XY} = \SET w \land X =  x')\cf Y=y'\notin\Gamma$; but then, by axiom C9 (and T1), since $\varphi_{End(Y)}\in \Gamma$,  $(\SET W_{XY} = \SET w \land X =  x')\cf Y=\F_Y(\SET wx)\in\Gamma$. In other words, for all $\SET w$ there is a value $y_{\SET w}$ such that, for \emph{all} values $x\in \ran(X)$, $(\SET W_{XY} = \SET w \land X =  x)\cf Y=y_{\SET w}\in\Gamma$. By the construction of $\T$, this tells us that $X$ is a dummy argument of $\F_Y$, i.e. $X\notin \PA_Y$.

Secondly, we observe that $\varphi_{DC(X,Y)}\vdash \aff{X}{Y}$. Indeed, $\varphi_{DC(X,Y)}$ is a disjunction of \emph{some} of the disjuncts of $\aff{X}{Y}$, so we derive $\aff{X}{Y}$ from it by T2. 



Suppose now that $\T$ is not recursive, i.e. there are variables $X_1,\dots, X_n$ ($n\geq 2$) such that $X_i\in \PA_{X_{i+1}}$ (for $i=1,\dots,n-1$) and $X_n \in \PA_1$. 
Then by the above we have $\aff{X_1}{X_2},\dots,\aff{X_{n-1}}{X_n},\aff{X_n}{X_1}\in\Gamma$. On the other hand, from $\aff{X_1}{X_2},\dots,\aff{X_{n-1}}{X_n}\in\Gamma$ we obtain, by axiom C11, that $(\aff{X_n}{X_1})^C\in\Gamma$, so we contradict the consistency of $\Gamma$. 


3) Since, by definition of $\T$, $\Pr(\Al_i)=\epsilon_i\in \Gamma$ for each $i\in I$, by Lemma \ref{lemma: decomposed probabilities} $\Pr(\bigvee_{i=1..n}\Al_i) = \sum_{i=1..n}\epsilon_i\in\Gamma$. On the other hand, by axioms A3 and P1 we have $\Pr(\bigvee_{i=1..n}\Al_i)=1\in\Gamma$. Thus, by Lemma \ref{lemma: properties of MCS}, 2. (unicity), $\sum_{i=1..n}\epsilon_i =1$.




4) By 3. we have $1 = \sum_{i=1..n}\epsilon_i = \sum_{i=1..n}\frac{m_i}{d}$, which implies $\sum_{i=1..n}m_i = d$. But then, since by definition $\T^-$ contains $m_i$ copies of each assignment $s_i$, its cardinality is $d$.

5) By definition, $P_\T(\Al_i)= \frac{\sum_{s\in\B_\sigma \mid(\{s\},\F)\models\Al_i}\#(s,\T)}{\T^-} = \frac{\#(s_i,\T)}{\T^-} = \frac{\#(s_i,\T)}{d} = \frac{m_i}{d} = \epsilon_i$, where in the third equality we have used 4.

\vspace{3pt}

6) We have $P_\T(\beta) = P_\T(\bigvee_{i \mid (s_i,\F)\models\beta}\Al_i) = \sum_{i \mid (s_i,\F)\models\beta}P_\T(\Al_i) = \sum_{i \mid (s_i,\F)\models\beta}\epsilon_i$, where in the last equality we used 5.
\end{proof}




\begin{lemma}\label{lemma: Psi F}
Writing $\Gamma$ and $\T = (\T^-,\F)$ as above, we have $\Phi^\F \in \Gamma$.
\end{lemma}

\begin{proof}
 Notice that, if $V\in End(\F)$, 
 then, by definition of $\T$, $\varphi_{End(V)}\in\Gamma$, and furthermore each of the conjuncts of $\eta_\sigma(V)$ is in $\Gamma$.

If instead $V\in Exo(\F)$, by definition of $\T$, $\varphi_{End(V)}\notin\Gamma$. Thus, by maximality of $\Gamma$, $(\varphi_{End(V)})^C\in\Gamma$, So, by axiom C10, $V=v \supset (\SET W_V=\SET w \cf V=v)$ for each $v$ and $\SET w$. Thus $\xi_\sigma(V)$ is in $\Gamma$.
%
%
\end{proof}

\begin{lemma}\label{lemma: C13}
Let $\beta\in\CO_\sigma$ and the $\hat \alpha_i$ formulas be defined as above. For any $\epsilon\in[0,1]\cap\mathbb Q$ and any function component $\F$ of signature $\sigma$, we have:
\[
 \vdash \Phi^\F \rightarrow (\Pr(\beta)=\epsilon \leftrightarrow \Pr\left(\bigvee_{i \mid (\{s_i(i/Key)\},\F)\models\beta}\Al_i\right)=\epsilon).
\]  
\end{lemma}
\begin{proof}
 Notice that the formula in the statement immediately follows from $\vdash \Phi^\F \rightarrow (\beta \equiv \bigvee_{i \mid (\{s_i(i/Key)\},\F)\models\beta}\Al_i)$ by axiom P1, P6 and Rule Rep. Since this latter is a $\CO$ formula, by Theorem \ref{thm: CO completeness} it suffices to show that it is semantically valid.   

 Let $T=(T^-,\G)$ be a causal multiteam of signature $\sigma$, and suppose $T\models \Phi^\F$. Then (see remarks in section \ref{subs: further notation}) $\G=\F$. By flatness of $\CO$, we just need to prove that, if $t\in T^-$, then $(\{t\},\F) \models \beta$ iff $(\{t\},\F) \models \bigvee_{i \mid (\{s_i(i/Key)\},\F)\models\beta}\Al_i$. But each $t\in T^-$ is of the form $s_i(k/Key)$ for some $i\in\{1,\dots,n\}$ and $k\in \mathbb N$. Then, obviously, $(\{t\},\F) \models \beta$ iff $(\{s_i(i/Key)\},\F) \models \beta$, since the value of $Key$ does not affect the satisfaction of formulas. But $(\{s_i(i/Key)\},\F) \models \Al_i$ by the definition of $\Al_i$; thus $(\{t\},\F) \models \beta$ iff $(\{s_i(i/Key)\},\F) \models \Al_i$ iff $(\{s_i(i/Key)\},\F) \models \bigvee_{i \mid (\{s_i(i/Key)\},\F)\models\beta} \Al_i$; and finally, as before, iff $(\{t\},\F) \models \bigvee_{i \mid (\{s_i(i/Key)\},\F)\models\beta} \Al_i$.
\end{proof}

\begin{lemma}\label{lemma: conditional probabilities in canonical T}
Let $\alpha,\beta\in\CO_{\sigma}$, $\delta\in[0,1]\cap\mathbb Q$, and $\T,\Gamma$ as before.
\begin{enumerate}
\item $P_\T(\beta) = \delta$ iff $\Pr(\beta)=\delta \in \Gamma$.
\item If $P_\T(\beta\mid\alpha)=\delta$, then $\alpha\supset\Pr(\beta)=\delta\in\Gamma$.
\item If $\alpha\supset\Pr(\beta)=\delta\in\Gamma$, then $T^\alpha$ is empty or $P_\T(\beta\mid\alpha)=\delta$.
\end{enumerate}
\end{lemma}

\begin{proof}

2) Notice that $P_\T(\alpha) =  \sum_{i\mid(s_i,\F)\models \alpha} \epsilon_i$. By 
Lemma \ref{lemma: decomposed probabilities} we have $\Pr(\bigvee_{i \mid (s_i,\F)\models \alpha}\Al_i) =  \sum_{i\mid(s_i,\F)\models \alpha} \epsilon_i\in\Gamma$. 
Thus by Lemma \ref{lemma: Psi F} and Lemma \ref{lemma: C13} we obtain (1): $\Pr(\alpha) =  \sum_{i\mid(s_i,\F)\models \alpha} \epsilon_i\in\Gamma$.
Similarly, we obtain (2): $\Pr(\alpha\land\beta) =  \sum_{i\mid(s_i,\F)\models \alpha\land\beta} \epsilon_i\in\Gamma$.

Notice furthermore the right-hand side term in (1) is $>0$ by the assumption that $P_\T(\cdot \mid \alpha)$ is defined. Thus we can apply axiom O2 to (1),(2) to obtain
\[
\alpha \supset \Pr(\beta) = \frac{ \sum_{i\mid(s_i,\F)\models \alpha\land\beta} \epsilon_i}{\sum_{i\mid(s_i,\F)\models \alpha} \epsilon_i} \in\Gamma.
\]
But the number occurring in this formula is just $\frac{P_\T(\alpha\land\beta)}{P_\T(\alpha)} = P_\T(\beta \mid \alpha)$ by Lemma \ref{lemma: properties of canonical T}, 5.

3) By Lemma \ref{lemma: properties of MCS}, 3., we obtain that there is a unique $\delta'$ such that $\Pr(\alpha)=\delta'\in\Gamma$.  By Lemma \ref{lemma: Psi F} and Lemma \ref{lemma: C13} we obtain $\Pr(\bigvee_{i \mid (s_i,\F)\models \alpha}\Al_i)=\delta'\in\Gamma$. By 
Lemma \ref{lemma: decomposed probabilities}, again, we also obtain that $\Pr(\bigvee_{i \mid (s_i,\F)\models \alpha}\Al_i)=\sum_{i\mid(s_i,\F)\models \alpha} \epsilon_i\in\Gamma$. Thus, again by Lemma \ref{lemma: properties of MCS}, 3. (unicity), $\delta'=\sum_{i\mid(s_i,\F)\models \alpha} \epsilon_i$.

  Suppose first that $\delta'=0$. Then $\sum_{i\mid(s_i,\F)\models \alpha} \epsilon_i=0$ and thus $\epsilon_i=0$ for any $i$ such that $(s_i,\F)\models \alpha$. In other words, by the construction of $\T^-$, $\#(s_i,\T)=0$. Thus $T^\alpha$ is empty.

Suppose instead that $\delta'>0$. Then by Lemma \ref{lemma: reverse axioms}, 1., we have $\Pr(\alpha)>0\in\Gamma$. But then by Lemma \ref{lemma: properties of MCS}, 6., there is a unique $\delta$ such that $\alpha \supset \Pr(\beta)=\delta\in\Gamma$. Suppose first that $\delta=0$; then, since $\Pr(\alpha)>0\in\Gamma$, by axiom O2 and Lemma \ref{lemma: properties of MCS}, 3. it must also be that $\Pr(\alpha\land\beta)=0\in\Gamma$. Then, $0 = \sum_{i\mid(s_i,\F)\models \alpha\land\beta} \epsilon_i$ by Lemmas \ref{lemma: Psi F}, \ref{lemma: C13} and \ref{lemma: decomposed probabilities}, and so $P_\T(\alpha\land\beta)=0$ by Lemma \ref{lemma: properties of canonical T}, 6.  Thus, 
\[
P_\T(\beta \mid \alpha) = \frac{P_\T(\alpha \land \beta)}{P_\T(\alpha)} = \frac{\sum_{i\mid(s_i,\F)\models \alpha\land\beta} \epsilon_i}{P_\T(\alpha)} = 0 = \delta.
\]

If instead $\delta>0$, we can apply axiom O3 to obtain $\Pr(\alpha \land \beta) = \delta\cdot\delta'\in\Gamma$. By Lemma \ref{lemma: Psi F} and Lemma \ref{lemma: C13} we have then $\Pr(\bigvee_{i \mid (s_i,\F)\models\alpha \land \beta}\Al_i) = \delta\cdot\delta'\in\Gamma$. Reasoning as before, we then get $\delta\cdot\delta' = \sum_{i\mid(s_i,\F)\models \alpha\land\beta} \epsilon_i$. But then
\[
P_\T(\beta \mid \alpha) = \frac{P_\T(\alpha \land \beta)}{P_\T(\alpha)} = \frac{\sum_{i\mid(s_i,\F)\models \alpha\land\beta} \epsilon_i}{\sum_{i\mid(s_i,\F)\models \alpha} \epsilon_i} = \frac{\delta\cdot\delta'}{\delta'} = \delta,
\]
where in the second equality we used Lemma \ref{lemma: properties of canonical T}, 6.

1) By 2. and 3. we have $P_\T(\beta) = P_\T(\beta\mid\top)=\delta$ iff $\top\supset \Pr(\beta)=\delta \in\Gamma$. However, we can show that $\top\supset \Pr(\beta)=\delta \in\Gamma$ iff $\Pr(\beta)=\delta \in\Gamma$, and then we are done.

Assume $\top\supset \Pr(\beta)=\delta \in\Gamma$. We also have $\top\in\Gamma$. Thus, by Theorem \ref{lemma: MP for supset}, 1., $\Pr(\beta)=\delta \in\Gamma$.


Vice versa, let us assume that $\Pr(\beta)=\delta \in\Gamma$. By T2, $\beta\equiv(\top\land\beta)\in\Gamma$; thus, $\Pr(\beta\equiv(\top\land\beta))=1\in\Gamma$ by P1; and then, by axiom P6 and Rule MP, (*): $\Pr(\top\land\beta)=\delta\in\Gamma$. 
On the other hand, by T1, $\top\in\Gamma$; then, by P1, (**): $\Pr(\top)=1\in\Gamma$.
Then, by axiom O2 and rule MP applied to (**) and (*), $\top\supset\Pr(\beta)=\delta \in\Gamma$.
\end{proof}

\subsection{Key lemmas and conclusion}

We will use the fact that every $\PCO$ formula can be put in a normal form (which is also closer to the formalism used in causal inference).

\begin{theorem}[Pearl-style normal form]\label{thm: Pearl normal form}
Every $\PCO_{\sigma}$ formula $\varphi$ is provably equivalent to a $\PCO_{\sigma}$ formula $\varphi'$ such that:
\begin{enumerate}
\item all consequents of $\cf$ are probabilistic atoms
\item all consequents of $\supset$ are counterfactuals or probabilistic atoms.
\end{enumerate}
\end{theorem}

\begin{proof}
1. is obtained by first pushing $\cf$ inwards using axioms C1-2-3-4-4b  (or removing occurrences of $\cf$ by axiom C5) and finally replacing all consequents of the form $X=x$ (resp. $X\neq x$) with $\Pr(X=x)\geq 1$ (resp. $\Pr(X\neq x)\geq 1$) by axioms P1 and T1. Each of these steps may require the use of Rule Rep. 

After achieving 1., we apply axioms O5$_\land$-O5$_\sqcup$-O5$_\supset$, together with Rule Rep, 
in order to push $\supset$ inwards, until it is in front of a counterfactual or an atom (if the latter is not probabilistic, replace it using axioms P1 and T1, as above). 
\end{proof}

\begin{lemma}[Truth lemma]\label{lemma: truth lemma}
Let $\Gamma$ and $\T$ be as before.
\begin{enumerate}
\item For all $\alpha\in\CO$ and $\varphi\in\PCO$, $\T^\alpha \models \varphi \iff \alpha\supset\varphi\in\Gamma$.
\item For all $\varphi\in\PCO$, $\T \models \varphi \iff \varphi\in\Gamma$.
\end{enumerate}
\end{lemma}

\begin{proof}
Notice first that 2. follows from 1. Indeed, 1. yields in particular that, for all  $\varphi\in\PCO$, $\T^\top \models \varphi \iff \top\supset\varphi\in\Gamma$. Now, on one hand, $\T^\top = \T$. On the other hand, $\top\supset\varphi\in\Gamma$ iff $\varphi\in\Gamma$, as shown in the proof of \ref{lemma: conditional probabilities in canonical T},1.

Let us then prove 1.; we will proceed by induction on $\varphi$. By Theorem \ref{thm: Pearl normal form}, we can assume that all counterfactuals in it are of the form $\SET X = \SET x \cf \Pr(\beta)\vartriangleright\epsilon$ ($\vartriangleright=\geq$ or $>$) or $\SET X = \SET x \cf \Pr(\beta)\vartriangleright \Pr(\gamma)$, and that $\supset$ only occurs in front of counterfactuals or probabilistic atoms. Furthermore, we can apply axiom C8 and C8b (together with the Rep Rule) 
to transform subformulas of the form $\SET X = \SET x \cf \Pr(\beta)\vartriangleright\epsilon$ into $\Pr(\SET X = \SET x \cf \beta)\vartriangleright\epsilon$, and subformulas of the form $\SET X = \SET x \cf \Pr(\beta)\vartriangleright\Pr(\gamma)$ into $\Pr(\SET X = \SET x \cf \beta)\vartriangleright \Pr(\SET X = \SET x \cf \gamma)$. We thus have just 4 base cases for the induction, $\varphi$ of the forms $\Pr(\beta)\vartriangleright\epsilon$ or $\Pr(\beta)\vartriangleright \Pr(\gamma)$; and we have no inductive case for $\cf$.
\begin{itemize}
\item Base case: $\varphi$ is $\Pr(\beta)\vartriangleright\epsilon$. Assume $\alpha\supset\varphi\in\Gamma$. Notice that, if $\vartriangleright$ is $>$, we obtain anyways $\alpha\supset \Pr(\beta)\geq\epsilon \in\Gamma$ by Lemma \ref{lemma: reverse axioms}, 2., and Rule Mon$_\supset$.
By Lemma \ref{lemma: properties of MCS}, 7., there is a $\delta\geq \epsilon$ such that $\alpha\supset \Pr(\beta)=\delta \in\Gamma$. By Lemma \ref{lemma: conditional probabilities in canonical T}, 3., we have then that either $\T^\alpha$ is empty or that $P_\T(\beta \mid \alpha) =\delta$. The former possibility immediately implies $\T^\alpha\models \Pr(\beta)\vartriangleright\epsilon$ by the empty team property. The latter implies $P_\T(\beta \mid \alpha) \geq \epsilon$ and thus we are done in case $\vartriangleright$ is $\geq$. Otherwise, if $\vartriangleright$ is $>$, 
observe that if $\delta>\epsilon$, we are done; otherwise, if $\delta=\epsilon$, $\alpha\supset \Pr(\beta)=\epsilon \in\Gamma$ and $\alpha\supset \Pr(\beta)>\epsilon \in\Gamma$  lead by T1 and Rule Mon$_\supset$ to $\alpha\supset \bot \in\Gamma$, and thus by O1b to $\Pr(\alpha)=0\in\Gamma$. By Lemma \ref{lemma: properties of canonical T}, 6., this implies that $\sum_{i \mid (s_i,\F)\models\alpha}\epsilon_i = 0$, i.e. $\T^\alpha$ is empty and again $\T^\alpha\models \Pr(\beta)\vartriangleright\epsilon$ by the empty team property. 



Vice versa, assume that $\T^\alpha\models \Pr(\beta) \vartriangleright \epsilon$. So, either $\T^\alpha$ is empty or $P_\T(\beta\mid\alpha)\vartriangleright \epsilon$. In the first case, we have $P_\T(\alpha) = \sum_{i | (s_i,\F)\models \alpha} P_\T(\Al_i) =0$. Then, $P_\T(\Al_i) =0$ for each such $i$, and thus, by the definition of $\T^-$, $\Pr(\Al_i)=0\in\Gamma$. Thus, as in the proof of Lemma \ref{lemma: properties of canonical T}, 3., we obtain $\Pr(\bigvee_{i | (s_i,\F)\models \alpha}\Al_i)=0\in\Gamma$  and then, by Lemma \ref{lemma: Psi F} and Lemma \ref{lemma: C13}, $\Pr(\alpha)=0\in\Gamma$. Thus, by axiom O1, $\alpha\supset \Pr(\beta)\vartriangleright \epsilon\in \Gamma$, as needed.

If $P_\T(\beta\mid\alpha)\vartriangleright \epsilon$, instead, then there is a $\delta\vartriangleright \epsilon$ such that  $P_\T(\beta\mid\alpha) = \delta$. Thus by Lemma  \ref{lemma: conditional probabilities in canonical T}, 2., $\alpha\supset \Pr(\beta)=\delta\in\Gamma$. Thus by Lemma \ref{lemma: 8 11 12}, 3. (and 1., in case $\vartriangleright = >$)  we get  $\alpha\supset \Pr(\beta)\vartriangleright \epsilon\in \Gamma$, as needed.

\item Base case: $\varphi$ is $\Pr(\beta)\vartriangleright \Pr(\gamma)$. Suppose $\alpha \supset \Pr(\beta)\vartriangleright \Pr(\gamma)\in\Gamma$.   
Assume first that $\Pr(\alpha)= 0\in\Gamma$. Then, by Lemma \ref{lemma: conditional probabilities in canonical T}, 1.,  $P_\T(\alpha)=0$. Thus $\T^\alpha$ is empty, and thus $\T^\alpha \models \Pr(\beta)\vartriangleright \Pr(\gamma)$. If instead $\Pr(\alpha)= 0\notin\Gamma$,
by Lemma \ref{lemma: properties of MCS}, 3., there is an $\epsilon>0$ such that $\Pr(\alpha)= \epsilon\in\Gamma$. Thus, by Lemma \ref{lemma: conditional probabilities in canonical T}, 1., $P_\T(\alpha)=\epsilon>0$, i.e., $\T^\alpha$ is nonempty.
On the other hand, by the maximality of $\Gamma$ we have $\Pr(\alpha)\neq 0\in\Gamma$, and thus, by P2 and T1, $\Pr(\alpha) >0\in\Gamma$. Thus by Lemma \ref{lemma: properties of MCS}, 6., there are (unique) $\delta, \delta'$ such that (*): $\alpha\supset\Pr(\beta)=\delta\in\Gamma$ and $\alpha\supset\Pr(\gamma)=\delta'\in\Gamma$. Then, by Lemma \ref{lemma: conditional probabilities in canonical T}, 3., since $\T^\alpha$ is not empty, we obtain $P_\T(\beta \mid \alpha) = \delta$ and  $P_\T(\gamma \mid \alpha) = \delta'$. So, if we prove $\delta\vartriangleright \delta'$, we are done. Suppose first that $\vartriangleright$ is $\geq$, and suppose for the sake of contradiction that $\delta < \delta'$. From (*), then, by axiom CP2 and Rule Mon$_\supset$ we get $\alpha\supset \Pr(\gamma) > \Pr(\beta)\in \Gamma$. Since $\Pr(\gamma) > \Pr(\beta)$ is $(\Pr(\beta)\geq \Pr(\gamma))^C$, we obtain $\alpha\supset\bot\in\Gamma$ by T1 and Rule Mon$_\supset$, and then, by axiom O1b, $\Pr(\alpha) = 0 \in \Gamma$, contradicting our assumption. If instead $\vartriangleright$ is $>$, suppose for the sake of contradiction that $\delta \leq \delta'$. Then by axiom CP1 and Rule Mon$_\supset$ we obtain $\alpha\supset \Pr(\gamma) \geq \Pr(\beta)\in \Gamma$, and then we proceed as before.

Vice versa, assume $\T\models\alpha\supset \Pr(\beta) \vartriangleright \Pr(\gamma)$. Then, either $\T^\alpha$ is empty or $\delta:= P_\T(\beta\mid\alpha) \vartriangleright P_\T(\gamma\mid\alpha) =:\delta'$. In the first case, by Lemma \ref{lemma: conditional probabilities in canonical T}, 1., we have $\Pr(\alpha)=0\in\Gamma$. Thus $\alpha\supset \Pr(\beta)\vartriangleright\Pr(\gamma)\in\Gamma$ by axiom O1. In the second case, by Lemma \ref{lemma: conditional probabilities in canonical T}, 2., we have that $\alpha\supset \Pr(\beta)=\delta$ and $\alpha\supset \Pr(\gamma)=\delta'$ are in $\Gamma$. Thus, by axiom CP1 (or CP2) and Rule Mon$_\supset$, $\alpha\supset \Pr(\beta)\vartriangleright\Pr(\gamma)\in\Gamma$.


\item Case $\varphi$ is $\psi\land\chi$. Then $\T\models\alpha\supset\varphi$ iff $\T^\alpha\models \psi\land\chi$ iff $\T^\alpha\models \psi$ and $\T^\alpha\models \chi$ iff (by inductive hypothesis) $\alpha\supset\psi\in\Gamma$ and $\alpha\supset\chi\in\Gamma$. This is equivalent to $\alpha\supset(\psi\land\chi)\in\Gamma$ by T1 and O5$_\land$. 
 
\item Case $\varphi$ is $\psi\sqcup\chi$. This is completely analogous to the previous case, using T1 and O5$_\sqcup$. 

\item Case $\varphi$ is $\alpha' \supset \chi$. Then $\T\models \alpha\supset\varphi$ iff $(\T^\alpha)^{\alpha'}\models\chi$ iff $\T^{\alpha\land\alpha'}\models\chi$ iff (by the inductive assumption) $(\alpha\land\alpha')\supset\chi\in\Gamma$. By axiom O5$_\supset$ 
the latter is equivalent to $\alpha \supset (\alpha'\supset\chi)\in \Gamma$, i.e. $\alpha\supset\varphi\in\Gamma$. 
\end{itemize}
\end{proof}

\begin{theorem}[Strong completeness]\label{thm: strong completeness of PCO}
 Let $\sigma$ be a signature and $\Gamma\cup\{\varphi\}\subseteq \PCO_{\sigma}$.
 Then $\Gamma \models \varphi$ implies $\Gamma\vdash \varphi$.
\end{theorem}

\begin{proof}
Suppose $\Gamma\not\vdash \varphi$. Then $\Gamma,\varphi^C\not\vdash\bot$ (otherwise by the deduction theorem $\Gamma\vdash \varphi^C\rightarrow\bot$, and then, by axiom T1 and Rule MP, $\Gamma\vdash\varphi$); in other words, $\Gamma \cup\{\varphi^C\}$ is consistent. Thus by Lemma \ref{lemma: Lindenbaum} there is a maximal consistent set $\Gamma'\supseteq \Gamma \cup\{\varphi^C\}$. Let $\T$ be the canonical causal multiteam built from $\Gamma'$ (which is a recursive causal multiteam by Lemma \ref{lemma: properties of canonical T}). By the truth Lemma \ref{lemma: truth lemma}, $\T\models \Gamma$ and $\T\models \varphi^C$. Thus $\T\models \Gamma$ and $\T\not\models \varphi$, so that $\Gamma\not\models\varphi$.
\end{proof}

\section{Conclusions}
We produced a strongly complete axiom system for a language  $\PCO$ for probabilistic counterfactual reasoning (without arithmetical operations).
As for most analogous results in the literature on interventionist counterfactuals, we have assumed that the signatures are finite; it would be interesting to find out if the recently developed methods of \cite{HalPet2022} for axiomatizatizing  infinite signatures may be extended to our case. Our system features infinitary rules, and it is therefore natural to wonder whether finitary axiomatizations could be obtained. Due to the failure of compactness, such axiomatizations can aspire at most at weak completeness.

There is another closely related axiomatization is issue that would be important to settle. In \cite{BarSan2023}, an extension $\PCO^\omega$ of $\PCO$ is considered that features a countably infinite version of the global disjunction $\sqcup$. This uncountable language is much more expressive than $\PCO$ and it can be proved that, in a sense, it encompasses all the expressive resources that a probabilistic language for interventionist counterfactuals should have. Given the special semantic role of this language, it would be important to find out whether an (infinitary) strongly complete axiomatization can be obtained for it. Proving a Lindenbaum lemma for an uncountable language and an infinitary axiom system presents peculiar difficulties, as discussed e.g. in \cite{BilCinLav2018}.

\section*{Acknowledgments}
Fausto Barbero was partially supported by the DFG grant VI 1045-1/1 and by the Academy of Finland grant 349803.
Jonni Virtema was partially supported by the DFG grant VI 1045-1/1 and by the Academy of Finland grant 338259.

\bibliographystyle{apalike}
\bibliography{iilogics}

\end{document}